\documentclass{daj}

\dajAUTHORdetails{%
  title = {Geometric Rank of Tensors and Subrank of Matrix Multiplication}, 
  author = {Swastik Kopparty, Guy Moshkovitz, and Jeroen Zuiddam},
  plaintextauthor = {Swastik Kopparty, Guy Moshkovitz, Jeroen Zuiddam},
  keywords = {algebraic complexity theory, combinatorics, matrix multiplication, tensors, subrank, analytic rank, slice rank, hypergraphs, independence number},
}   

\dajEDITORdetails{%
   year={2023},
   number={1},
   received={23 July 2020},   
   revised={6 January 2022},    
   published={27 April 2023},  
   doi={10.19086/da.73322},       
}   

\usepackage{mathtools, amsthm, amsfonts, amssymb, commath}
\usepackage[capitalize]{cleveref} 
\usepackage{accents} 

\usepackage{microtype}
\usepackage{fancyvrb}

\usepackage{enumitem}

\theoremstyle{plain}
\newtheorem{theorem}{Theorem}[section]
\newtheorem{main}{Theorem}
\newtheorem*{theorem*}{Theorem}

\newtheorem{lemma}[theorem]{Lemma}

\theoremstyle{definition}
\newtheorem{definition}[theorem]{Definition}

\newtheorem*{problem*}{Problem}

\newtheorem{remark}[theorem]{Remark}
\newtheorem{example}{Example}

\newcommand{\defin}[1]{\emph{#1}}

\renewcommand{\b}{\beta}

\newcommand{\eps}{\varepsilon}

\newcommand{\D}{\Delta}
\newcommand{\NN}{\mathbb{N}}
\newcommand{\FF}{\mathbb{F}}
\newcommand{\CC}{\mathbb{C}}
\newcommand{\ZZ}{\mathbb{Z}}
\newcommand{\KK}{\mathbb{K}}

\renewcommand{\b}[1]{\mathbf{#1}}
\newcommand{\sub}{\subseteq}

\newcommand{\F}{\mathbb{F}}
\newcommand{\K}{\mathbb{K}}

\newcommand{\N}{\mathbb{N}}

\DeclareMathOperator{\V}{\mathbf{V}}

\DeclareMathOperator*{\EE}{\mathbb{E}}

\DeclareMathOperator{\bias}{bias}

\DeclareMathOperator{\codim}{codim}
\DeclareMathOperator{\I}{I}
\DeclareMathOperator{\SR}{SR}
\DeclareMathOperator{\PR}{PR}

\DeclareMathOperator{\ZR}{ZR}

\DeclareMathOperator{\GR}{GR}
\DeclareMathOperator{\AR}{AR}
\DeclareMathOperator{\matrixrank}{rank}
\DeclareMathOperator{\corank}{corank}
\DeclareMathOperator{\Tr}{Tr}

\newcommand{\GL}{\mathrm{GL}}

\DeclarePairedDelimiter\ceil{\lceil}{\rceil}
\DeclarePairedDelimiter\floor{\lfloor}{\rfloor}

\newcommand{\degenleq}{\unlhd}
\newcommand{\id}{\mathrm{Id}}
\DeclareMathOperator{\rank}{R}
\DeclareMathOperator{\subrank}{Q}
\DeclareMathOperator{\bordersubrank}{\underline{Q}}
\DeclareMathAccent{\wtilde}{\mathord}{largesymbols}{"65}
\newcommand{\regularize}[1]{\underaccent{\wtilde}{#1}}

\DeclareMathOperator{\asympsubrank}{\underaccent{\wtilde}{Q}}

\DeclareMathOperator{\PP}{\mathbb{P}}

\begin{document}

\begin{frontmatter}[classification=text]

\title{Geometric Rank of Tensors and Subrank of Matrix Multiplication} 

\author[kopp]{Swastik Kopparty\thanks{Research supported in part by NSF grants CCF-1253886, CCF-1540634, CCF-1814409 and CCF-1412958, and BSF grant 2014359. Some of this research was done while visiting the Institute for Advanced Study.}}
\author[moshk]{Guy Moshkovitz\thanks{This work was conducted at the Institute for Advanced Study, enabled through support from the National Science Foundation under grant number CCF-1412958, and at DIMACS, enabled through support from the National Science Foundation under grant number CCF-1445755.}}
\author[zuid]{Jeroen Zuiddam\thanks{This work was conducted at the Institute for Advanced Study, supported by NSF grants DMS-1638352 and CCF-1900460, and at New York University, supported by a Simons Junior Fellowship. Any opinions, findings and conclusions or recommendations expressed in this material are those of the author and do not necessarily reflect the views of the National Science Foundation.}}

\begin{abstract}
		Motivated by problems in algebraic complexity theory (e.g., matrix multiplication) and extremal combinatorics (e.g., the cap set problem and the sunflower problem),
		we introduce the geometric rank as a new tool in the study of tensors and hypergraphs.
		We prove that the geometric rank is an upper bound on the subrank of tensors and the independence number of hypergraphs.
		We prove that the geometric rank is smaller than the slice rank of Tao, and relate geometric rank to the analytic rank of Gowers and Wolf in an asymptotic fashion.  
		As a first application, we use geometric rank to prove a tight upper bound on the (border) subrank of the matrix multiplication tensors, matching Strassen's well-known lower bound from 1987.
\end{abstract}
\end{frontmatter}





\section{Introduction}
	
	Tensors play a central role in computer science and mathematics.
	Motivated by problems in algebraic complexity theory (e.g.,~the arithmetic complexity of matrix multiplication), extremal combinatorics (e.g.,~the cap set problem and the Erd\H{o}s--Szemerédi sunflower problem) and quantum information theory (the resource theory of quantum entanglement), we introduce and study a new tensor parameter called {\em geometric~rank}.
	Like the many widely studied notions of rank for tensors (rank, subrank, border rank, border subrank, flattening rank, slice rank, analytic rank),
	geometric rank of tensors generalizes the classical rank of matrices. In this paper, we:
	\begin{itemize}
		\item prove a number of basic properties and invariances of geometric rank,
		\item develop several tools to reason about, and sometimes exactly compute, the geometric rank,
		\item show intimate connections between geometric rank and the other important notions of rank for tensors,
		\item and as a simple application of the above, we answer an old question of Strassen by showing that the (border) subrank of $m \times m$ matrix multiplication is at most $\lceil 3m^2/4 \rceil$
		(this is tight for border subrank; previously the border subrank of the matrix multiplication tensor was known to lie between $\frac{3}{4}m^2$ and $(1 - o(1))m^2$).
	\end{itemize}
	More generally, we believe that geometric rank provides an interesting new route to prove upper bounds on subrank of tensors (and hence independence numbers of hypergraphs).
	Such upper bounds are important in complexity theory in the context of matrix multiplication and barriers to matrix multiplication, and combinatorics in the context
	of specific natural hypergraphs (as in the cap set problem and the Erd\H{o}s--Szemerédi sunflower problem).
	Indeed, we prove that in these kinds of applications the geometric rank performs as good as the established slice rank. By means of the aforementioned application to matrix multiplication we see that the geometric rank can indeed perform better than slice rank and is moreover a \emph{precise} method (in the sense that its value in that case precisely coincides with the lower bound on the border subrank that was obtained by Strassen).

	
	\subsection{Geometric rank}
	We define the geometric rank of a tensor as the codimension of the (possibly reducible) algebraic variety defined by the bilinear forms given by the slices of the tensor. Here we use the standard notions of dimension and codimension of affine varieties from algebraic geometry. That is, for any tensor $T = (T_{i,j,k})_{i,j,k} \in \FF^{n_1 \times n_2 \times n_3}$ with coefficients $T_{i,j,k}$ in an algebraically closed field $\FF$ (e.g.,~the complex numbers $\CC$) and with 3-slices~$M_k = (T_{i,j,k})_{i,j} \in \FF^{n_1 \times n_2}$
	we define the geometric rank~$\GR(T)$ as
	\[
	\GR(T) = \codim \{ (x,y) \in \FF^{n_1}\times\FF^{n_2} \mid x^T\! M_1 y = \cdots = x^T\! M_{n_3} y = 0\}.
	\]
	%
	Viewing $T$ as the trilinear map $T : \FF^{n_1} \times \FF^{n_2} \times \FF^{n_3} \to \FF : (x,y,z) \mapsto \sum_{i,j,k} T_{i,j,k}\, x_i y_j z_k$, we can equivalently write the geometric rank of $T$ as
	\[
	\GR(T) = \codim \{ (x,y) \in \FF^{n_1} \times \FF^{n_2} \mid \forall z\in \FF^{n_3} : T(x,y,z) = 0\}.
	\]
	The definition of geometric rank is expressed asymmetrically in $x$, $y$ and $z$. We will see, however, that the codimensions of $\{ (x,y)\in \FF^{n_1}\times\FF^{n_2} \mid \forall z : T(x,y,z) = 0\}$, $\{ (x,z)\in \FF^{n_1}\times\FF^{n_3} \mid \forall y : T(x,y,z) = 0\}$ and $\{ (y,z)\in \FF^{n_2}\times\FF^{n_3} \mid \forall x : T(x,y,z) = 0\}$ coincide (\cref{symmetry}).
	
	The motivation for this definition is a bit hard to explain right away. We arrived at it while searching for a characteristic $0$ analogue of the analytic rank of Gowers and Wolf~\cite{MR2773103}
	(see \cref{sec:ARGR}).
	
	\begin{example}\label{ex1}
		We give an example of how to compute the geometric rank. Let $T \in \FF^{2\times 2 \times 2}$ be the tensor with 3-slices
		\[
		M_1 = \begin{pmatrix} 1 & 0\\ 0 & 0\end{pmatrix},\quad M_2 = \begin{pmatrix} 0 & 1\\1 & 0\end{pmatrix}.
		\]
		(This is sometimes called the $W$-tensor).
		One verifies that the algebraic variety $V =
		\{(x,y) \in \FF^2 \times \FF^2 \mid x_1 y_1 = 0, x_2y_1 + x_1y_2 = 0\}$
		has the three irreducible components $\{(x,y) \in \FF^2 \times \FF^2 \mid x_1 = 0, x_2 = 0\}$, $\{(x,y) \in \FF^2 \times \FF^2 \mid x_1 = 0, y_1 = 0\}$ and $\{(x,y) \in \FF^2 \times \FF^2 \mid y_1 = 0, y_2 = 0\}$.
		Each irreducible component has dimension 2
		and thus $V$ has dimension 2. Hence $\GR(T) = \codim V = 4 - 2 = 2$.
		%
		We will see more examples of geometric rank later (\cref{stropt}).
	\end{example}

	\subsection{Overview: notions of tensor rank}
	
	Before discussing our results we give an introduction to some of the existing notions of rank of tensors and their usefulness.
	Several interesting notions of rank of tensors have been studied in mathematics and computer science, each with their own applications. 
	%
	%
	As a warm-up we first discuss the familiar situation for matrices (which are tensors of order two).
	
	\paragraph{Matrices.}
	For any two matrices $M \in \FF^{m_1 \times m_2}$ and~$N \in \FF^{n_1 \times n_2}$ we write $M \leq N$ if there exist matrices $A,B$ such that $M = ANB$.
	Defining the matrix rank $\rank(M)$ of $M$ as the smallest number~$r$ such that $M$ can be written as a sum of~$r$ matrices that are outer products $(u_i v_j)_{ij}$ (i.e., rank-$1$ matrices),
	we see that in terms of the relation $\leq$ we can write the matrix rank as the minimisation
	\[
	\rank(M) = \min \{r \in \NN \mid M \leq I_r\},
	\]
	where $I_r$ is the $r\times r$ identity matrix. Matrix rank thus measures the ``cost'' of $M$ in terms of identity matrices. Let us define the subrank $\subrank(M)$ of $M$ as the ``value'' of $M$ in terms of identity matrices,
	\[
	\subrank(M) = \max \{s \in \NN \mid I_s \leq M\}.
	\]
	It turns out that subrank equals rank for matrices,
	\[
	\subrank(M) = \rank(M).
	\]
	Namely, if $\rank(M) = r$, then by using Gaussian elimination we can bring $M$ in diagonal form with exactly~$r$ nonzero elements on the diagonal, and so $I_r \leq M$.
	In fact, $M \leq N$ if and only if~$\rank(M) \leq \rank(N)$. 
	
	\paragraph{Tensors.}
	We will now focus on tensors of order three. (All our results in fact hold for tensors of arbitrary order. We stick to order three for simplicity, and because the proofs for order three and arbitrary order are essentially identical.)
	For any two tensors $S \in \FF^{m_1 \times m_2 \times m_3}$ and $T \in \FF^{n_1 \times n_2 \times n_3}$ we write~$S \leq T$ if there are matrices $A,B,C$ such that $S = (A,B,C)\cdot T$ where we define $(A,B,C) \cdot T \coloneqq (\sum_{a,b,c} A_{ia} B_{jb} C_{kc} T_{a,b,c})_{i,j,k}$. Thus $(A,B,C)\cdot T$ denotes taking linear combinations of the slices of $T$ in three directions according to $A$, $B$ and $C$.
	Let $T \in \FF^{n_1 \times n_2 \times n_3}$ be a tensor.
	The tensor rank $\rank(T)$ of $T$ is defined as the smallest number~$r$ such that $T$ can be written as a sum of $r$ tensors that are outer products~$(u_i v_j w_k)_{i,j,k}$.
	Similarly as for matrices, we can write tensor rank in terms of the relation $\leq$ as the ``cost'' minimisation
	\[
	\rank(T) = \min \{r \in \NN \mid M \leq I_r\}
	\]
	where $I_r$ is the $r \times r \times r$ identity tensor (i.e., the diagonal tensor with ones on the main diagonal).
	Strassen defined the subrank of $T$ as the ``value'' of $T$ in terms of identity tensors,
	\[
	\subrank(T) = \max \{s \in \NN \mid I_s \leq M\}.
	\]
	Naturally, since $\leq$ is transitive, we have that value is at most cost: $\subrank(T) \leq \rank(T)$.
	Unlike the situation for matrices, however, 
	there exist tensors for which this inequality is strict. 
	One way to see this is using the fact that a random tensor in $\FF^{n\times n\times n}$ has tensor rank close to $n^2$ whereas its subrank is at most $n$.
	Another way to see this is using the flattening ranks. These are defined as the ranks $\rank^{(i)}(T) \coloneqq \rank(T^{(i)})$ of the (``flattening'') matrices
	$T^{(1)} = (T_{i,j,k})_{i,(j,k)} \in \FF^{n_1 \times n_2n_3}$, $T^{(2)} = (T_{i,j,k})_{j,(i,k)} \in \FF^{n_2 \times n_1n_3}$, and $T^{(3)} = (T_{i,j,k})_{k,(i,j)} \in \FF^{n_3 \times n_1n_2}$ obtained from $T$ by grouping two of the three indices together. For these it holds that
	\[
	\subrank(T) \leq \rank^{(i)}(T) \leq \rank(T).
	\]
	Indeed, it is not hard to find tensors $T$ for which $\rank^{(1)}(T) < \rank^{(2)}(T)$. (Then, in particular, we have a strict inequality $\subrank(T) < \rank(T)$.)

	We will now discuss several methods to upper bound the subrank $\subrank(T)$ that improve on the flattening ranks~$\rank^{(i)}(T)$. Then in \cref{subsec:conn} we will discuss connections between subrank and problems in complexity theory and combinatorics.
	%
	%
	
	\paragraph{Slice rank.}
	In the context of the cap set problem, Tao~\cite{tao} defined the slice rank of any tensor~$T$ as the minimum number $r$ such that $T$ can be written as a sum of $r$ tensors of the form $(u_i V_{jk})_{i,j,k}$, $(u_j V_{ik})_{i,j,k}$ or $(u_k V_{ij})_{i,j,k}$ (i.e., an outer product of a vector and a matrix). In other words,  $\SR(T) \coloneqq \min \{ \rank^{(1)}(S_1) + \rank^{(2)}(S_2) + \rank^{(3)}(S_3) : S_1 + S_2 + S_3 = T\}$. Clearly, slice rank is at most the flattening rank~$\rank^{(i)}$ for any $i$. Tao proved that slice rank still upper bounds subrank,
	\[
	\subrank(T) \leq \SR(T) \leq \rank^{(i)}(T).
	\]
	The lower bound connects slice rank to problems in extremal combinatorics, which we will discuss further in \cref{subsec:conn}.
	The slice rank of large Kronecker powers of tensors was studied in
	\cite{MR3631613} and \cite{MR4495838}, which lead to strong connections with invariant theory and moment polytopes, and with the asymptotic spectrum of tensors introduced by Strassen \cite{strassen1988asymptotic}.
	
	\paragraph{Analytic rank.}
	Gowers and Wolf \cite{MR2773103} defined the analytic rank of any tensor $T \in \FF_p^{n_1 \times n_2 \times n_3}$ over the finite field~$\FF_p$ for a prime $p$
	as $\AR(T) \coloneqq -\log_p \bias(T)$, where the bias of $T$ is defined as $\bias(T) \coloneqq \EE \exp(2\pi i\, T(x,y,z)/p)$ with the expectation taken over all vectors $x \in \FF_p^{n_1}$, $y \in \FF_p^{n_2}$ and~$z \in \FF_p^{n_3}$.
	The analytic rank relates to subrank and tensor rank as follows:
	\[
	\subrank(T) \leq \frac{\AR(T)}{\AR(I_1)} \leq \rank(T)
	\]
	where $\AR(I_1) = -\log_p (1 - (1-1/p)^{2})$. The upper bound was proven in \cite{bhrushundi_et_al:LIPIcs:2020:12632}.
	Interestingly, the value of $\AR(T)/\AR(I_1)$ can be larger than~$\max_i \rank^{(i)}(T)$ for small $p$ (leading to non-trivial lower bounds on the tensor rank \cite{bhrushundi_et_al:LIPIcs:2020:12632}).
	The lower bound is essentially by Lovett~\cite{lovett2019analytic}.
	Namely, Lovett proves that $\AR(T)/\AR(I_1)$ upper bounds the size of the largest principal subtensor of $T$ that is diagonal. (We will discuss this further in \cref{subsec:conn}.)
	Lovett moreover proved that analytic rank is at most slice rank, $\AR(T) \leq \SR(T)$, and thus proposes analytic rank as an effective upper bound method for any type of problem where slice rank works well asymptotically.
	Lovett's result motivated us to study other parameters to upper bound the subrank, which led to geometric rank.
	
	%
	
	\paragraph{Higher-order tensors and partition rank.}

	So far we have discussed tensors of order three (and two). Now we will say something about $k$-tensors for general $k$. The notions of tensor rank, subrank and analytic rank generalize directly to this setting. Also slice rank generalizes directly as the minimum number $r$ such that $T$ can be written as a sum of $r$ tensors that are an outer product of a vectors and a $(k-1)$-tensor. In this higher-order situation there is another important tensor parameter, called partition rank, introduced by Naslund \cite{NASLUND2020105190} as a relaxation of slice rank (and used to solve problems in combinatorics). For~$k=3$, partition rank coincides with slice rank, but for $k>3$ partition rank can be strictly smaller than slice rank. The partition rank $\PR(T)$ of a $k$-tensor~$T$ is  defined as the minimum number $r$ such that $T$ can be written as a sum of $r$ tensors that are each an outer product of an $\ell$-tensor and a $(k-\ell)$-tensor, for arbitrary $1 \leq \ell \leq k-1$ (that can be different for each of the $r$ summands).

	A long line of work has shown upper bounds on $\PR(T)$ in terms of $\AR(T)$. This was first proven by Green and Tao \cite{DBLP:journals/cdm/GreenT09}, whose method was later improved on by Kaufman and Lovett \cite{DBLP:conf/focs/KaufmanL08} and by Bhowmick and Lovett \cite{bhowmick2015bias}. However, these results imply an upper bound on $\PR(T)$ that is an Ackermann-type function of $\AR(T)$. The dependence was later improved significantly by Janzer \cite{janzer2018low}. Recently, Janzer \cite{janzer2019polynomial} and Mili{\'c}evi{\'c} \cite{milicevic2019polynomial} proved polynomial upper bounds on $\PR$ in terms of $\AR$. Even more recently, a \emph{linear} bound was given in \cite{cohen2021partition} under the assumption that the finite field is of mildly large size (independent of the size $n$ of the tensor).


	\subsection{Connections of subrank to complexity theory and combinatorics}\label{subsec:conn}

	\paragraph{Arithmetic complexity of matrix multiplication and barriers.}
	A well-known problem in computer science concerning tensors is about the arithmetic complexity of matrix multiplication. Asymptotically how many scalar additions and multiplications are required to multiply two $m\times m$ matrices? The answer is known to be between $n^2$ and $C n^{2.37...}$, or in other words, the exponent of matrix multiplication $\omega$ is known to be between~$2$ and $2.37...$~\cite{le2014powers}.  The complexity of matrix multiplication turns out to be determined by the tensor rank of the matrix multiplication tensors~$\langle m,m,m\rangle$ corresponding to taking the trace of the product of three $m\times m$ matrices. Explicitly, $\langle m,m,m\rangle$ corresponds to the trilinear map~$\sum_{i,j,k = 1}^m x_{ij}y_{jk}z_{ki}$.
	In practice, upper bounds on the rank of the matrix multiplication tensors are obtained by proving a chain of inequalities
	\[
	\langle m,m,m\rangle \leq T \leq I_r
	\]
	for some intermediate tensor $T$, which is usually taken to be a Coppersmith--Winograd tensor~\cite{MR1056627}, and an $r$ that is small relatively to $m$.
	It was first shown by Ambainis, Filmus and Le Gall \cite{MR3388238} that there is a barrier for this strategy to give fast algorithms. This barrier was recently extended and simplified in several works
	\cite{MR3631613, blasiak2017groups, 8555139, DBLP:conf/coco/Alman19, MR4322894} and can be roughly phrased as follows: if the asymptotic subrank of the intermediate tensor $\lim_{n\to\infty} \subrank(T^{\otimes n})^{1/n}$ is strictly smaller than the asymptotic rank $\lim_{n\to \infty} \rank(T^{\otimes n})^{1/n}$, then one cannot obtain $\omega = 2$ via~$T$. These barriers rely on the fact that the asymptotic subrank of the matrix multiplication tensors   is maximal. Summarizing, the rank of the matrix multiplication tensors corresponds to the complexity of matrix multiplication whereas the subrank of any tensor corresponds to the a priori suitability of that tensor for use as an intermediate tensor. The upper bounds on the asymptotic subrank used in the aforementioned results were obtained via slice rank or the related theory of support functionals and quantum functionals~\cite{MR4495838}.
	
	\paragraph{Cap sets, sunflowers and independent sets in hypergraphs.}
	
	Several well-known problems in extremal combinatorics can be phrased in terms of the independence number of families of hypergraphs. One effective collection of upper bound methods proceeds via the subrank of tensors. (For other upper bound methods, see e.g.~the recent work of Filmus, Golubev and Lifshitz~\cite{filmus2019high}.) A hypergraph is a a symmetric subset $E \subseteq V \times V \times V$.
	An independent set of $E$ is any subset $S\subseteq V$ such that $S$ does not induce any edges in $E$, that is, $E \cap (S \times S \times S) = \emptyset$. The independence number~$\alpha(E)$ of $E$ is the largest size of any independent set in $E$. For any hypergraph $E \subseteq [n] \times [n] \times [n]$, if $T \subseteq \FF^{n \times n \times n}$ is any tensor supported on $E \cup \{(i,i,i) : i \in [n]\}$, then
	\[
	\alpha(E) \leq \subrank(T).
	\]
	Indeed, for any independent set $S$ of $E$ the subtensor $T|_{S \times S \times S}$ is a diagonal tensor with nonzero diagonal and $T \geq T|_{S \times S \times S}$.
	For example, the resolution of the cap set problem by Ellenberg and Gijswijt~\cite{MR3583358}, as simplified by Tao~\cite{tao}, can be thought of as upper bounding the subrank of tensors corresponding to strong powers of the hypergraph consisting of the edge $(1,2,3)$ and permutations. The Erd\H{o}s--Szemer\'edi sunflower problem for three petals was resolved by Naslund and Sawin \cite{naslund2017upper} by similarly considering the strong powers of the hypergraph consisting of the edge $(1,1,2)$ and permutations. In both cases slice rank was used to obtain the upper bound.
	Another result in extremal combinatorics via analytic rank was recently obtained by Bri\"et~\cite{briet2019subspaces}.
	
	
	\subsection{Our results}
	
	We establish a number of basic properties of geometric rank. These imply close connections between geometric rank and
	other notions of rank, and thus bring in a new set of algebraic geometric tools to help reason about the
	various notions of rank. In particular, our new upper bounds on the (border) subrank of matrix multiplication follow easily from our basic
	results.
	
	\paragraph{Subrank and slice rank.}
	We prove that the geometric rank $\GR(T)$ is at most the slice rank~$\SR(T)$ of Tao \cite{tao} and at least the subrank~$\subrank(T)$ of Strassen \cite{strassen1987relative} (see \cref{theo:basic-bounds}).
	\begin{main}\label{theo:basic-bounds-intro}
		For any tensor $T$,
		\[
		\subrank(T) \le \GR(T) \le \SR(T).
		\]
	\end{main}
	We thus add $\GR$ to the collection of methods to upper bound the subrank of tensors~$\subrank$ and in turn the independence number of hypergraphs. \cref{theo:basic-bounds-intro} in particular says that in any application of slice rank to upper bound an independence number, the geometric rank provides an upper bound that is at least as good (and indeed sometimes better).
	We prove these inequalities by proving that $\GR$ is monotone under~$\leq$, additive under the direct sum of tensors, and has value $1$ on the trivial $I_1$ tensor.
	We also give a second more direct proof of this inequality (\cref{moredirect}).
	
	\paragraph{Border subrank.}
	We extend our upper bound on subrank to {\em border subrank}, the (widely studied) approximative version of subrank.
	
	The main ingredient in this extension is the following fact (which itself exploits the algebraic-geometric nature of definition of $\GR$): the set of tensors $\{T \in \FF^{n \times n \times n} \mid \GR(T) \leq m\}$ is closed in the Zariski topology.\footnote{That is, the statement $\GR(T) \leq m$ is characterized by the vanishing of a finite number of polynomials.}
	In other words, geometric rank is lower-semicontinuous. This implies that the geometric rank also upper bounds the border subrank~$\bordersubrank(T)$ (see \cref{border}).
	\begin{main}\label{border-intro}
		For any tensor $T$,
		\[
		\bordersubrank(T) \leq \GR(T).
		\]
	\end{main}
	The geometric rank $\GR$ is a new parameter in the study of tensors. Indeed we show that $\GR$ is not the same  parameter as the subrank~$\subrank$, border subrank~$\bordersubrank$ or slice rank $\SR$ (\cref{notqbar} and \cref{notsr}).\footnote{More broadly, it has been brought to our attention that Schmidt in \cite[Sec.~16]{MR781588} and \cite[Sec.~6]{MR774104} studied a parameter for homogenous polynomials that is similar to our geometric rank (with an extra polarization step).}
	
	\paragraph{Matrix multiplication.}
	In the study of the complexity of matrix multiplication, Strassen~\cite{strassen1987relative} proved that for the matrix multiplication tensors $\langle m,m,m\rangle \in \FF^{m^2 \times m^2 \times m^2}$ the border subrank is lower bounded by $\lceil\tfrac34 m^2\rceil \leq \bordersubrank(\langle m,m,m\rangle)$. We prove that this lower bound is optimal by proving
	the following (see \cref{stropt}). 
	\begin{main}\label{stropt-intro}
		For any positive integers $e \leq h \leq \ell$,
		\[
		\bordersubrank(\langle e,h,\ell\rangle) = \GR(\langle e,h,\ell\rangle) = \begin{cases}
		eh - \floor{\frac{(e + h - \ell)^2}{4}} & \textnormal{if } e + h \geq \ell,\\
		eh & \textnormal{otherwise}.
		\end{cases}
		\]
		In particular,
		we have $\subrank(\langle m,m,m \rangle) \leq \bordersubrank(\langle m,m,m\rangle) = \GR(\langle m,m,m\rangle) =\ceil{\tfrac34m^2}$ for any $m \in \NN$.
	\end{main}
	Our computation of $\GR$ here is a calculation of the dimension of a variety. We do this by studying the dimension of various sections of that
	variety, which then reduces to linear algebraic questions about matrices (we are talking about matrix multiplication after all).

	Our result improves the previously best known upper bound on the subrank of matrix multiplication of Christandl, Lucia, Vrana and Werner~\cite{christ2018tensor}, which was $\subrank(\langle m,m,m\rangle) \leq m^2 - m$. In fact, our upper bound on~$\GR(\langle e,h,\ell\rangle)$ exactly matches the lower bound on~$\bordersubrank(\langle e, h, \ell\rangle)$ of Strassen~\cite{strassen1987relative}, for any nonnegative integers $e$, $h$, and $\ell$. We thus solve the problem of determining the exact value of $\bordersubrank(\langle e, h, \ell\rangle)$.
	
	Not only does \cref{stropt-intro} provide an instance where geometric rank performs strictly better than slice rank (i.e. is smaller), it also shows that the geometric rank is a precise method, since in this case the value precisely coincides with the border subrank.
	
	\paragraph{Analytic rank.}
	Finally, we establish a strong connection between geometric rank and analytic rank.
	
	We prove that for any tensor $T \in \ZZ^{n_1 \times n_2 \times n_3} \subseteq \CC^{n_1 \times n_2 \times n_3}$ with integer coefficients, the geometric rank of $T$ equals the {\em liminf} of the analytic rank of the tensors~$T_p \in \FF_p^{n_1 \times n_2 \times n_3}$ obtained from~$T$ by reducing all coefficients modulo $p$ and letting $p$ go to infinity over all primes (see \cref{theo:liminf}).
	\begin{main}
		For every tensor $T$ over $\ZZ$ we have 
		\[
		\liminf_{p\to\infty} \AR(T_p) = \GR(T).
		\]
	\end{main}
	This result is in fact the source of our definition of geometric rank. The analytic rank of a tensor is defined as the bias of a certain polynomial on random inputs. By simple transformations,
	computing the analytic rank over $\F_p$ reduces to computing the number of solutions of a system of polynomial equations over $\F_p$. Namely,
	\begin{equation}\label{eq:AR-identity}
	\AR(T_p) = n_1 + n_2 -\log_p \abs[0]{ \{ (x,y) \in \FF_p^{n_1} \times \FF_p^{n_2} : T_p(x,y,\cdot)=0 \}}.
	\end{equation}
	This system of polynomial equations defines a variety,
	and it is natural to expect that the dimension of the variety roughly determines the number of $\F_p$-points of the variety. This expectation is not true in general, but under highly controlled circumstances
	something like it is true. This is how we arrived at the definition of geometric rank (which eventually turned out to have very natural properties on its own, without this connection to analytic rank).
	Actually establishing the above liminf result is quite roundabout, and requires a number of tools from algebraic geometry and number theory. 
	
	We stress that analytic rank is only defined for tensors over prime fields of positive characteristic, whereas geometric rank is defined for tensors over any field. By the aforementioned result, geometric rank over the complex numbers can be thought of as an extension of analytic rank to characteristic~0. Finding an extension of analytic rank beyond finite fields is mentioned as an open problem by Lovett~\cite[Problem 1.10]{lovett2019analytic}.
	
	\subsection{Higher-order tensors}\label{subsec:ktensors}
	
	All definitions and theorems (and their proofs) stated above for tensors of order three naturally generalize to tensors of arbitrary order. For concreteness we state the higher-order versions here explicitly. 

	We begin with the higher-order definition. For any tensor $T$  with coefficients $T_{i_1,i_2,\ldots,i_k}$ in an algebraically closed field $\FF$
	we define the geometric rank~$\GR(T)$ as follows.
	%
	Viewing $T$ as the $k$-linear map $T : \FF^{n_1} \times \FF^{n_2}\times \cdots \times \FF^{n_k} \to \FF : (x_1,x_2, \ldots, x_k) \mapsto \sum_{i_1,i_2,\ldots, i_k} T_{i_1, i_2, \ldots, i_k}\, (x_1)_{i_1} (x_2)_{i_2} \cdots (x_k)_{i_k}$, the geometric rank of $T$ is defined as
	\[
	\GR(T) = \codim \{ (x_1, x_2, \ldots, x_{k-1}) \in \FF^{n_1} \times \FF^{n_2} \times \cdots \times \FF^{n_{k-1}} \mid \forall x_k \in \FF^{n_k} : T(x_1,x_2,\ldots, x_k) = 0\}.
	\]
	
	Also on $k$-tensors $\GR$ is monotone under $\leq$, additive under the direct sum of tensors, and has value 1 on the trivial $I_1$ tensor. This leads to a proof of the following theorem, which extends \cref{theo:basic-bounds-intro} and \cref{border-intro} from 3-tensors to $k$-tensors for any $k$. Note how slice rank is here replaced by the (stronger) partition rank.

	\begin{main}
		For any $k$-tensor $T$,
		\[
		\subrank(T) \leq \bordersubrank(T) \le \GR(T) \le \PR(T).
		\]
	\end{main}
	
	Finally, for $k$-tensors the following relation between analytic rank and geometric rank holds.

	\begin{main}
		For any $k$-tensor $T$ over $\ZZ$ we have 
		\[
		\liminf_{p\to\infty} \AR(T_p) = \GR(T).
		\]
	\end{main}

	\subsection*{Follow-up work}
	The recent work of Cohen and Moshkovitz in \cite{MR4492184} crucially uses geometric rank as an intermediary to prove a tight relation between the slice rank and the analytic rank of 3-tensors. Their later work extends this result to a similar relation between the partition rank and the analytic rank of $k$-tensors over large fields \cite{cohen2021partition}. 
	
	In other recent follow-up work to our paper, Geng and Landsberg~\cite{geng2021geometry} classify tensors with non-maximal geometric rank in $\CC^3 \otimes \CC^3 \otimes \CC^3$, classify tensors with geometric rank two, and show that upper bounds on geometric rank imply lower bounds on tensor rank. This line of work was further extended by Geng~\cite{https://doi.org/10.48550/arxiv.2201.03615}.

	\subsection*{Organization of this paper}
	
	In the next section we formally define geometric rank. In~\cref{sec:alternative}, we give
	some alternative definitions of geometric rank that help us reason about it.
	In~\cref{sec:basicbounds} and~\cref{sec:border} we show the relationship between geometric rank,
	slice rank, subrank and border subrank. In~\cref{sec:matmult} we use the established properties
	of geometric rank to give a proof of our upper bound on the (border) subrank of matrix multiplication.
	In~\cref{sec:moredirect} we give a more direct proof of the inequality between slice rank and geometric rank.
	Finally, in~\cref{sec:ARGR} we establish the relationship between geometric and analytic ranks.

	\section{Geometric rank}\label{sec:gr}
	In this section we set up some general notation and define geometric rank.
	Let $\FF$ be an algebraically closed field.  

	\paragraph{Dimension and codimension.} The notion of dimension that we use is the standard notion in algebraic geometry, and is defined as follows. Let~$V \subseteq \FF^n$ be a (possibly reducible) algebraic variety (i.e.~the common zero set of a set of polynomials). The codimension $\codim V$ is defined as $n - \dim V$.
	The dimension $\dim V$ is defined as the
	length of a maximal chain of irreducible subvarieties of~$V$~\cite{harris2013algebraic}. In our proofs we will use basic facts about dimension: the dimension of a linear space coincides with the notion from linear algebra, the dimension is additive under the cartesian product and the dimension of a locally closed set equals the dimension of its closure. We will moreover use a basic theorem about the dimension of fibers under projections $(x,y) \mapsto y$. 

	\paragraph{Notation about tensors.}
	Let $\FF^{n_1 \times n_2 \times n_3}$ be the set of all three-dimensional arrays
	\[
	T = (T_{i,j,k})_{i\in [n_1], j \in [n_2], k \in [n_3]}
	\]
	with $T_{i,j,k} \in \FF$. We refer to the elements of $\FF^{n_1 \times n_2 \times n_3}$ as the \defin{$n_1 \times n_2 \times n_3$ tensors over $\FF$}. To any tensor $T \in \FF^{n_1 \times n_2 \times n_3}$ we associate the polynomial in $\FF[x_1, \ldots, x_{n_1}, y_1, \ldots, y_{n_2}, z_1, \ldots, z_{n_3}]$ defined~by
	\[
	T(x_1, \ldots, x_{n_1}, y_1, \ldots, y_{n_2}, z_1, \ldots, z_{n_3}) = \sum_{i \in [n_1]} \sum_{j\in [n_2]} \sum_{k\in [n_3]} T_{i,j,k}\, x_i y_j z_k
	\]
	and the trilinear map $\FF^{n_1} \times \FF^{n_2} \times \FF^{n_3} \to \FF$ defined by
	\[
	T(x, y, z) = T(x_1, \ldots, x_{n_1}, y_1, \ldots, y_{n_2}, z_1, \ldots, z_{n_3}).
	\]
	

	\begin{definition}\label{grdef}
		The \emph{geometric rank} of a tensor $T\in \FF^{n_1 \times n_2 \times n_3}$, written $\GR(T)$, is the codimension of the set of elements $(x,y) \in \FF^{n_1} \times \FF^{n_2}$ such that $T(x,y,z) = 0$ for all $z\in \FF^{n_3}$. That is,
		%
		\[
		\GR(T) \coloneqq \codim \{ (x,y) \in \FF^{n_1} \times \FF^{n_2} \mid \forall z \in \FF^{n_3} : T(x, y, z) = 0\}.
		\]
		For any $(x,y) \in \FF^{n_2} \times \FF^{n_3}$ we define the vector $T(x,y,\cdot) = (T(x,y, e_k))_{k=1}^{n_3}$, where $e_1, \ldots, e_{n_3}$ is the standard basis of $\FF^{n_3}$.
		In this notation the geometric rank is given by
		\[
		\GR(T) = \codim \{ (x,y) \mid T(x, y, \cdot) = 0\}.
		\]
		For later use we also define the vectors $T(x, \cdot,z) = T(x,e_j, z)_{j}$ and $T(\cdot,y,z) = T(e_i, y, z)_{i}$, and
		we define the matrices $T(x,\cdot,\cdot) = T(x,e_j, e_k)_{j,k}$,\,  $T(\cdot, y,\cdot) = T(e_i, y, e_k)_{i,k}$ and $T(\cdot,\cdot,z) = T(e_i, e_j, z)_{i,j}$.
	\end{definition}
	We defined the geometric rank of tensors with coefficients in an algebraically closed field. For tensors with coefficients in an arbitrary field we naturally define the geometric rank via the embedding of the field in its algebraic closure.

	\paragraph{Computer software.}
	One can compute the dimension of an algebraic variety $V\subseteq \FF^n$ using computer software like Macaulay2~\cite{M2} or Sage~\cite{sagemath}. This allows us to easily compute the geometric rank of small tensors. For example, for \cref{ex1} in the introduction over the field~$\FF = \CC$, one verifies in Macaulay2 with the commands
	\begin{Verbatim}
	R = CC[x1,x2,y1,y2];
	dim ideal(x1*y1, x2*y1 + x1*y2)
	\end{Verbatim}
	or in Sage with the commands
	\begin{Verbatim}
	A.<x1,x2,y1,y2> = AffineSpace(4, CC);
	Ideal([x1*y1, x2*y1 + x1*y2]).dimension()
	\end{Verbatim}
	that $\dim V = 2$.
	
	\paragraph{Computational complexity.}
	Koiran~\cite{646091} studied the computational complexity (as $n$ grows) of the problem of deciding whether, given an algebraic variety $V\subseteq \CC^n$ and a number $d \in \NN$, the dimension of~$V$ is at least $d$. When $V$ is given by polynomial equations over the integers the problem is in PSPACE, and assuming the Generalized Riemann Hypothesis the problem is in the Arthur--Merlin class~AM. Thus the same upper bounds apply to computing $\GR$.
	In the other direction, Koiran showed that computing dimension of algebraic varieties in general is NP-hard. We know of no hardness results for computing~$\GR$.

	
	\paragraph{Higher-order tensors.}
	Our definition of geometric rank extends naturally from 3-tensors $\FF^{n_1 \times n_2 \times n_3}$ to 
	$k$-tensors $\FF^{n_1 \times \cdots \times n_k}$ for any $k\geq 2$ by defining the geometric rank of any $k$-tensor $T \in \FF^{n_1 \times \cdots \times n_k}$ as
	\[
	\GR(T) \coloneqq \codim \{ (x_1, \ldots, x_{k-1}) \in \FF^{n_1} \times \cdots \times \FF^{n_{k-1}} \mid \forall x_k \in \FF^{n_k} : T(x_1, \ldots, x_{k-1}, x_k) = 0\}.
	\]
	For $k=2$ geometric rank coincides with matrix rank.
	Our results extend naturally to $k$-tensors with this definition (see also \cref{subsec:ktensors}), but for clarity our exposition will be in terms of 3-tensors.

	\section{Alternative descriptions of geometric rank}
	\label{sec:alternative}
	We give two alternative descriptions of geometric rank that we will use later. The first description relates geometric rank to the matrix rank of the matrices $T(x,\cdot, \cdot) = (T(x,e_j, e_k))_{j,k}$. The second description (which follows from the first) shows that the geometric rank of $T(x,y,z)$ is invariant under permuting the variables~$x$,~$y$ and $z$.
	Both theorems rely on an understanding of the dimension of fibers of a (nice) map.

	
	

	\begin{theorem}\label{asmax} Let $T \in \FF^{n_1 \times n_2 \times n_3}$ be a tensor. Then
		\begin{enumerate}[label=\upshape(\roman*)]
		\item 
		$\dim \{ (x,y) \mid T(x,y,\cdot) = 0\} 
		=\max_i \left(\dim \left\{x \mid \corank T(x,\cdot,\cdot) = i\right\} + i\right).$
		\item 
		$\GR(T) =\min_j \left(\codim \left\{x \mid \matrixrank T(x,\cdot,\cdot) = j\right\} + j\right).$
		\end{enumerate}
	\end{theorem}
	\begin{proof}
		(i)
		Let $V = \{(x,y) \mid T(x,y,\cdot) = 0\}$. The goal is to reexpress $\dim V$ as in the claim. 
		Let $W = \FF^{n_1}$.
		Let $\pi : V \to W$ map $(x,y)$ to $x$.
		Define the sets $W_i = \{ x \mid \corank(T(x,\cdot,\cdot)) = i\}$.
		The rank-nullity theorem for matrices gives for any fixed $x$ that $\corank(T(x,\cdot,\cdot)) = \dim \{y \mid T(x,y,\cdot) = 0\}$.
		The sets $W_i$ are locally closed, that is, each $W_i$ is the intersection of an open set and a closed set.
		Let $V_i = \pi^{-1}(W_i)$. The set~$V_i$ is also locally closed.
		We have that $W = \cup_i W_i$ and so $V = \cup_i V_i$. Therefore, $\dim V = \max_i \dim V_i$.
		It follows from \cite[Proposition 10.6.1(iii)]{EGA-IV3} that $\dim V_i = \dim W_i + i$.
		From this claim follows $\dim V = \max_i (\dim W_i + i)$, which finishes the proof.
		%

		(ii) This is a simple consequence of (i). Namely, we have by definition that 
		\[
			\GR(T) = \codim \{ (x,y) \mid T(x,y,\cdot) = 0\} = n_1 + n_2 - \dim \{ (x,y) \mid T(x,y,\cdot) = 0\}.
		\]
		We apply claim (i) to find
		\begin{align*}
			n_1 + n_2 - \dim \{ (x,y) \mid T(x,y,\cdot) = 0\}
			= n_1 + n_2 - \max_i \bigl( \dim \{x : \corank T(x,\cdot,\cdot) = i\} + i \bigr).
		\end{align*}
		Taking the maximization outside and using the relation between dimension and codimension and the relation between rank and corank, we find
		\begin{align*}
		&n_1 + n_2 - \max_i \bigl( \dim \{x : \corank T(x,\cdot,\cdot) = i\} + i \bigr)\\
		 &= \min_i\bigl( n_1 + n_2 - \bigl( \dim \{x : \corank T(x,\cdot,\cdot) = i\} + i\bigr)\bigr)\\
		&= \min_i \bigl(n_1 - \dim \{x : \matrixrank T(x, \cdot, \cdot) = n_2 - i\} + n_2 - i\bigr)\\
		&= \min_j\bigl( \codim \{x : \matrixrank T(x,\cdot, \cdot) = j\} + j\bigr).
		\end{align*}
		This proves (ii).
	\end{proof}

	
	\begin{theorem}\label{symmetry} For any tensor $T$,
		\begin{align*}
		\GR(T) &= \codim \{(x,y) \mid T(x, y, \cdot)= 0\}\\
		&= \codim \{(x,z) \mid T(x, \cdot, z)=0 \}\\
		&= \codim \{(y,z) \mid T(\cdot, y, z)=0\}.
		\end{align*}
	\end{theorem}
	\begin{proof}
		We apply \cref{asmax} (ii) to $T$ and to $T$ after swapping $y$ and $z$ to get that the codimensions of $\{(x,y) \mid T(x,y,\cdot) = 0 \}$ and $\{(x,z) \mid T(x,\cdot,z) = 0 \}$ are equal to $\min_j \codim \{ x \mid \matrixrank T(x,\cdot,\cdot) = j \} + j$. This proves the first equality.
		The second equality is proven similarly.
		%
		%
	\end{proof}
	
	
	
	The above theorems and proofs (like all our proofs) generalize directly from 3-tensors to $k$-tensors $T$ for all $k\geq 3$. This gives for \cref{asmax} the generalization :
	\begin{enumerate}[label=\upshape(\roman*)]
		\item 
		$\dim \{ (x_1, \ldots, x_{k-1}) \mid T(x_1,\ldots, x_{k-1},\cdot) = 0\}\\[0.5em] 
		=\max_i \bigl(\dim \left\{(x_1, \ldots, x_{k-2}) \mid \corank T(x_1,\ldots, x_{k-2}, \cdot,\cdot) = i \right\} + i\bigr).$
		\item 
		$\GR(T) =\min_j \bigl(\codim \left\{(x_1, \ldots, x_{k-2}) \mid \matrixrank T(x_1, \ldots, x_{k-2},\cdot,\cdot) = j\right\} + j\bigr)$
		\end{enumerate}
		and for \cref{symmetry} the generalization
		\begin{align*}
			\GR(T) &= \codim \{ (x_2, x_3, \ldots, x_{k}) \mid T(\cdot, x_2, x_3, \ldots, x_{k}) = 0\}\\
			&= \codim \{ (x_1, x_3, \ldots, x_{k}) \mid T(x_1,\cdot,x_3,  \ldots, x_{k}) = 0\}\\
			& \hspace{0.5em}\vdots\\
			&= \codim \{ (x_1, x_2, \ldots, x_{k-1}) \mid T(x_1, x_2, \ldots, x_{k-1}, \cdot) = 0\}.
		\end{align*}
	
	\section{Geometric rank lies between subrank and slice rank}\label{sec:basicbounds}
	
	Recall that the subrank $\subrank(T)$ of $T$ is the largest number $s$ such that $I_s \leq T$ and the slice rank~$\SR(T)$ is the smallest number $r$ such that $T(x,y,z)$ can be written as a sum of $r$ trilinear maps of the form~$f(x)g(y,z)$ or $f(y)g(x,z)$ or $f(z)g(x,y)$.
	
	\begin{theorem}\label{theo:basic-bounds}
		For any tensor $T$,
		\[
		\subrank(T) \le \GR(T) \le \SR(T).
		\]
	\end{theorem}
	
	Theorem~\ref{theo:basic-bounds} will follow from the following basic properties of $\GR$. We will give a more direct proof of the inequality $\GR(T) \leq \SR(T)$ in \cref{sec:moredirect}. Recall from the introduction that for any two tensors $S \in \FF^{m_1 \times m_2 \times m_3}$ and $T \in \FF^{n_1 \times n_2 \times n_3}$ we write~$S \leq T$ if there are matrices $A,B,C$ such that $S = (A,B,C)\cdot T$ where we define $(A,B,C) \cdot T \coloneqq (\sum_{a,b,c} A_{ia} B_{jb} C_{kc} T_{a,b,c})_{i,j,k}$.
	
	\begin{lemma}\label{lemma:GR-monotone}
		$\GR$ is $\leq$-monotone: if $S \leq T$, then $\GR(S) \leq \GR(T)$.
	\end{lemma}
	
	\begin{proof} 
		
		Let $T \in \FF^{n_1 \times n_2 \times n_3}$. We claim that $\GR((\id, \id, C) \cdot T) \leq \GR(T)$ for any $C \in \FF^{m_3 \times n_3}$,
		where~$\id$ denotes an identity matrix of the appropriate size.
		From this claim and the symmetry of~$\GR$ (\cref{symmetry}), follows the inequalities $\GR((A,\id,\id)\cdot T) \leq T$ and $\GR((\id, B, \id)\cdot T) \leq \GR(T)$ for any matrices $A \in \FF^{m_1 \times n_1}$ and $B \in \FF^{m_2 \times n_2}$. Chaining these three inequalities gives that for any two tensors~$S$ and $T$, if $S \leq T$, then $\GR(S) \leq \GR(T)$.
		
		We prove the claim.
		Let $S = (\id, \id, C) \cdot T$.
		Let~$M_k = (T_{i,j,k})_{ij}$ be the 3-slices of~$T$ and let~$N_k = (S_{i,j,k})_{ij}$ be the 3-slices of~$S$.
		Since $S = (\id, \id, C) \cdot T$, the matrices $N_1,\ldots, N_{m_3}$ are in the linear span of the matrices $M_1, \ldots, M_{n_3}$.
		Thus $V = \{(x,y) \mid x^T M_1 y = \cdots = x^T M_{n_3} y = 0\}$ is a subset of $W = \{(x,y) \mid x^T N_1 y = \cdots = x^T N_{m_3} y = 0\}$. Therefore, $\dim V \leq \dim W$ and it follows that $\GR(S) = \codim W \leq \codim V = \GR(T)$.
	\end{proof}
	
	Let $T_1 \in \FF^{m_1 \times m_2 \times m_3}$ and $T_2 \in \FF^{n_1 \times n_2 \times n_3}$ be tensors with 3-slices $A_k = ((T_1)_{i,j,k})_{i,j}$ and $B_k = ((T_2)_{i,j,k})_{i,j}$ respectively.
	The direct sum $T_1 \oplus T_2 \in \FF^{(m_1 + n_1) \times (m_2 + n_2) \times (m_3 + n_3)}$
	is defined as the tensor with 3-slices
	$A_k \oplus 0_{n_1 \times n_2}$ for $k = 1, \ldots, m_3$ and
	$0_{m_1 \times m_2} \oplus B_k$
	for $k = m_3 + 1, \ldots, m_3 + n_3$ where $0_{a\times b}$ denotes the zero matrix of size $a\times b$. In other words, $T_1 \oplus T_2$ is the block-diagonal tensor with blocks $T_1$ and $T_2$.
	
	\begin{lemma}\label{lemma:GR-additive}
		$\GR$ is additive under direct sums: $\GR(T_1 \oplus T_2) = \GR(T_1)+\GR(T_2)$.
	\end{lemma}
	\begin{proof}
		Let $A_k$ be the 3-slices of $T_1$ and let $B_k$ be the 3-slices of~$T_2$.
		Let $T = T_1 \oplus T_2$ be the direct sum with 3-slices $M_k$.
		%
		Then
		\[
		V =
		\{(x,y) \mid T(x,y,\cdot) = 0\} = \{(x,y) \mid x^T M_1 y = \cdots = x^T M_{m_3+n_3} y = 0\}
		\]
		is the cartesian product of
		\[
		V_1 = \{(x,y) \mid x^T A_{1} y = \cdots = x^T A_{m_3} y = 0\}
		\]
		and
		\[
		V_2 = \{(x,y) \mid x^T B_{1} y = \cdots = x^T B_{n_3} y = 0\}.
		\]
		Thus $\dim V = \dim V_1 + \dim V_2$ \cite[page 138]{harris2013algebraic}. Therefore, 
		$\GR(T) = \GR(T_1)+\GR(T_2)$.
		%
		%
		%
	\end{proof}
	
	\begin{lemma}\label{lemma:GR-subadditive}
		$\GR$ is sub-additive under element-wise sums: $\GR(S+T) \le \GR(S)+\GR(T)$.
	\end{lemma}
	\begin{proof} 
		Note that $S+T \leq S \oplus T$. Thus,
		$\GR(S+T) \le \GR(S \oplus T) = \GR(S) + \GR(T)$,
		where the inequality uses Lemma~\ref{lemma:GR-monotone}, and the equality uses~\cref{lemma:GR-additive}.
	\end{proof}
	
	\begin{lemma}\label{lemma:arankslice}
		If $\SR(T) = 1$ then $\GR(T) = 1$.
	\end{lemma}
	\begin{proof}
		It is sufficient to consider a tensor $T \in \FF^{1 \times n \times n}$ with one nonzero slice. Then we have that $T(0, \FF^n, \FF^n) = 0$, and so $\GR(T) = 1 + n - n = 1$.
	\end{proof}
	
	We note that \cref{lemma:arankslice} extends to higher-order tensors with slice rank replaced by partition rank. That is, for any $k$-tensor $T$, if $\PR(T) = 1$ then $\GR(T) = 1$. The proof for that is again simple. If $T \in \FF^{n_1 \times \cdots \times n_k}$ is a $k$-tensor of partition rank one, then as a multilinear polynomial we can write (up to permuting the $k$ variable vectors $x_i$) $T(x_1, \ldots, x_k) = f(x_1, \ldots, x_\ell) g(x_{\ell+1}, \ldots, x_k)$ for some multilinear polynomials $f$ and $g$ and some $1 \leq \ell \leq k-1$.
	Then
	\begin{align*}
		&\codim \{(x_1, \ldots, x_{k-1}) \mid T(x_1, \ldots, x_{k-1}, \cdot) = 0\}\\
		&\leq \codim \{(x_1, \ldots, x_{k-1}) \mid f(x_1, \ldots, x_\ell) = 0\} = 1
	\end{align*}
	since $\{(x_1, \ldots, x_{k-1}) \mid f(x_1, \ldots, x_\ell) = 0\}$ is precisely cut out by one equation, namely $f(x_1, \ldots, x_\ell) = 0$.

	
	\begin{lemma}\label{lemma:GR-I}
		For every $r \in \NN$ we have
		$\GR(I_r)=r$.
	\end{lemma}
	\begin{proof}
		We have $\SR(I_1) = 1$ and so $\GR(I_1) = 1$ (\cref{lemma:arankslice}). Since $I_r$ is a direct sum of $r$ copies of $I_1$ and geometric rank is additive under taking the direct sum $\oplus$ (\cref{lemma:GR-subadditive}), we find $\GR(I_r) = r \GR(I_1) = r$.
	\end{proof}
	
	
	\begin{proof}[\bfseries\upshape Proof of Theorem~\ref{theo:basic-bounds}]
		We prove that $\GR(T) \leq \SR(T)$. Let $r=\SR(T)$. Then there are tensors $T_1, \ldots, T_r$ so that $T = \sum_{i=1}^r T_i$ and $\SR(T_i) = 1$. Then also $\GR(T_i) = 1$ (\cref{lemma:arankslice}).
		Subadditivity of $\GR$ under element-wise sums (\cref{lemma:GR-subadditive}) gives
		\[
		\GR(T) \le \sum_{i=1}^r \GR(T_i) = r = \SR(T).
		\]
		We prove that $\subrank(T) \leq \GR(T)$.
		Let $s=\subrank(T)$. Then $I_s \leq T$. We know $\GR(I_s) = s$ (\cref{lemma:GR-I}). By the $\leq$-monotonicity of $\GR$ (\cref{lemma:GR-monotone}), we have
		\[
		\subrank(T) = s = \GR(I_s) \leq \GR(T).\qedhere
		\]
	\end{proof}
	
	
	\section{Geometric rank is at least border subrank}\label{sec:border}
	
	In this section we extend the inequality $\subrank(T) \leq \GR(T)$ (\cref{theo:basic-bounds}) to the approximative version of subrank, called border subrank.
	To define border subrank we first define degeneration~$\degenleq$, which is the approximative version of restriction $\leq$.
	We write $S \degenleq T$, and we say $S$ is a \defin{degeneration} of $T$, if for some~$e\in \NN$ we have
	\[
	S + \eps S_1 + \eps^2 S_2 + \cdots + \eps^e S_e  = (A(\eps), B(\eps), C(\eps)) \cdot T
	\]
	for some tensors $S_i$ over $\FF$ and for some matrices $A(\eps), B(\eps), C(\eps)$ whose coefficients are Laurent polynomials in the formal variable $\eps$.
	%
	Equivalently, $S \degenleq T$ if and only if $S$ is in the orbit closure~$\overline{G \cdot T}$ where $G$ denotes the group $\GL_{n_1} \times \GL_{n_2} \times \GL_{n_3}$, $G\cdot T$ denotes the natural group action that we also used in the definition of $\leq$, and the closure is taken in the Zariski topology \cite[Theorem~20.24]{burgisser1997algebraic}. (When $\FF = \CC$ one may equivalently take the closure in the Euclidean topology.) Recall that the subrank of $T$ is defined as $\subrank(T) = \max \{n \in \NN \mid I_n \leq T\}$.
	The \defin{border subrank} of $T$ is defined as
	\[
	\bordersubrank(T) = \max \{n \in \NN \mid I_n \degenleq T\}.
	\]
	Clearly, $\subrank(T) \leq \bordersubrank(T)$.
	
	\begin{theorem}\label{border}
		For any tensor $T$,
		\[
		\bordersubrank(T) \leq \GR(T).
		\]
	\end{theorem}
	
	To prove \cref{border} 
	we use the following theorem on upper-semicontinuity of fiber dimension.
	
	
	\begin{theorem}[{\cite[special case of Corollary 11.13]{harris2013algebraic}}]\label{harris}
		Let $X$ be the zero set of bi-homogeneous polynomials, that is,
		\[
		X = \{(a,b) \in \FF^{m_1} \times \FF^{m_2} \mid f_1(a,b) = \cdots = f_k(a,b) = 0\}
		\]
		where the $f_i(a,b)$ are polynomials that are homogeneous in both $a$ and $b$.
		Let $\pi: X \to \FF^{m_2}$ map~$(a,b)$ to $b$.
		Let $Y = \pi(X)$ be its image.
		For any $q \in Y$, let $\lambda(q) = \dim(\pi^{-1}(q))$.
		Then $\lambda(q)$ is an upper-semicontinuous function of $q$, that is, the set $\{q \in Y \mid \lambda(q) \geq m\}$ is Zariski closed in $Y$.
	\end{theorem}

	\begin{lemma}\label{aranklim} $\GR$ is lower-semicontinuous: for any $n_i,m \in \NN$ the set $\{T \in \FF^{n_1 \times n_2 \times n_3} \mid \GR(T) \leq m\}$ is Zariski closed.
	\end{lemma}
	\begin{proof}
		We define the set
		\[
		X = \{(T,x,y) \in \FF^{n_1 \times n_2 \times n_3} \times \FF^{n_1} \times \FF^{n_2} \mid T(x,y, \cdot) = 0 \}.
		\]
		Let $\pi : X \to \FF^{n_1 \times n_2 \times n_3}$ map $(T,x,y)$ to $T$.
		Let $Y = \pi(X) = \FF^{n_1 \times n_2 \times n_3}$ be the image of $\pi$.
		For any~$T \in Y$ let $\lambda(T) \coloneqq \dim(\pi^{-1}(T))$.
		Then $\lambda(T)$ is an upper-semicontinuous function of $T$ in the Zariski topology on $Y$ by \cref{harris}.
		This means that the set $\{T \in \FF^{n_1 \times n_2 \times n_3} \mid \lambda(T) \geq m\}$ is closed for every $m \in \NN$. It follows that $\{ T \in \FF^{n_1 \times n_2 \times n_3} \mid \GR(T) \leq m \}$ is closed for every $m \in \NN$.
	\end{proof}
	
	\begin{remark}
		A well-known example of a lower-semicontinuous function is matrix rank.
		Indeed, the set of matrices of rank at most $m$ is the zero set of the determinants of all $(m+1) \times (m+1)$ submatrices. For geometric rank we do not know an explicit set of generators for the vanishing ideal of $\{T \in \FF^{n_1 \times n_2 \times n_3} \mid \GR(T) \leq m\}$. For slice rank the set $\{T \in \FF^{n_1 \times n_2 \times n_3} \mid \SR(T) \leq m\}$ is also known to be Zariski closed and explicit vanishing polynomials for this variety were recently obtained by Bläser, Ikenmeyer, Lysikov, Pandey and Schreyer~\cite{blser2019variety}.
	\end{remark}
	
	\begin{lemma}\label{arankdegenmon}
		$\GR$ is $\degenleq$-monotone: if $S\degenleq T$, then $\GR(S) \leq \GR(T)$.
	\end{lemma}
	\begin{proof}
		For all $g \in G$ we have $\GR(g\cdot T) = \GR(T)$ by \cref{lemma:GR-monotone}.
		The set $\{T' \mid \GR(T') \leq \GR(T)\}$ is Zariski closed by \cref{aranklim}.
		It contains the orbit $G\cdot T$ and hence also its Zariski closure~$\overline{G\cdot T}$, that is,
		\[
		\{T' \mid T' \degenleq T \} = \overline{G \cdot T} \subseteq \{T' \mid \GR(T') \le \GR(T) \}.
		\]
		Therefore, $\GR(S) \leq \GR(T)$.
	\end{proof}
	
	\begin{proof}[\bfseries\upshape Proof of \cref{border}]
		Let $n = \bordersubrank(T)$. Then $I_n \degenleq T$ by the definition of $\bordersubrank$, and so $n \leq \GR(T)$ by \cref{arankdegenmon}. This proves the claim.
	\end{proof}

	\section{The border subrank of matrix multiplication}\label{sec:matmult}
	
	In the context of constructing fast matrix multiplication algorithms, Strassen \cite[Theorem 6.6]{strassen1987relative} proved that for any positive integers $e \leq h \leq \ell$ the border subrank of the matrix multiplication tensor $\langle e,h,\ell\rangle$ is lower bounded by
	\begin{equation}\label{strineq}
	\bordersubrank(\langle e,h,\ell\rangle) \geq \begin{cases}
	eh - \floor{\frac{(e + h - \ell)^2}{4}} & \textnormal{if } e + h \geq \ell,\\
	eh & \textnormal{otherwise}.
	\end{cases}
	\end{equation}
	Here $\langle e,h,\ell\rangle$ is the tensor that corresponds to taking the trace of the product of an $e\times h$ matrix, an $h \times \ell$ matrix and an $\ell \times e$ matrix.
	We prove using the geometric rank that this lower bound is optimal.
	
	\begin{theorem}\label{stropt}
		For any positive integers $e \leq h \leq \ell$
		\[
		\bordersubrank(\langle e,h,\ell\rangle) = \GR(\langle e,h,\ell\rangle) = \begin{cases}
		eh - \floor{\frac{(e + h - \ell)^2}{4}} & \textnormal{if } e + h \geq \ell,\\
		eh & \textnormal{otherwise}.
		\end{cases}
		\]
		In particular,
		we have $\bordersubrank(\langle m,m,m\rangle) = \GR(\langle m,m,m\rangle) =\ceil{\tfrac34m^2}$ for any $m \in \NN$.
	\end{theorem}
	\begin{proof}
		Since $\bordersubrank(\langle e,h,\ell\rangle) \leq \GR(\langle e,h,\ell\rangle)$ (\cref{border}) and since we have the lower bound in \eqref{strineq}, it suffices to show that $\GR(\langle e,h,\ell\rangle)$ is at most $eh - \floor{(e + h - \ell)^2/4}$ if $e + h \geq \ell$ and at most $eh$ otherwise.
		
		Let $T = \langle e, h, \ell\rangle$.
		Let $V = \{(x,y) \in \FF^{e h} \times \FF^{h\ell} \mid T(x,y,\cdot)=0\}$.
		Then $\GR(T) = eh + h\ell - \dim V$.
		From \cref{asmax} it follows that
		\begin{equation}\label{Vmax}
		\dim V = \max_i \dim \{x \in \FF^{eh} \mid \dim \{y \in \FF^{h\ell} \mid T(x,y,\cdot) = 0\} = i\} + i.
		\end{equation}
		We now think of $\FF^{eh}$, $\FF^{h\ell}$ and $\FF^{\ell e}$ as the matrix spaces $\FF^{e \times h}$, $\FF^{h \times \ell}$ and $\FF^{\ell \times e}$. Then $T$ gives the trilinear map
		$T : \FF^{e \times h} \times \FF^{h \times \ell} \times \FF^{\ell \times e}  \to \FF : (X,Y,Z) \mapsto \Tr(XYZ)$.
		Therefore, $T(X,Y,\cdot) = 0$ if and only if $XY = 0$.
		If the rank of $X$ as an $e \times h$ matrix equals $r$, then
		\[\dim \{Y \in  \FF^{h \times \ell} \mid T(X,Y,\cdot) =  0\} = (h - r)\ell,\]
		since $Y$ is any matrix with columns from $\ker(X)$.
		We have
		\[\dim\{X \in \FF^{e \times h} \mid \matrixrank(X) = r\} = er + (h-r)r,\]
		since rank $r$ matrices can be parametrized (up to permutation of the columns) by choosing the first~$r$ columns independently, which contributes $er$ parameters, and choosing the last $h-r$ columns as linear combinations of the first $r$ columns, which contributes $(h-r)r$ parameters.
		Thus the relevant values of $i$ in \eqref{Vmax} are of the form $i = (h-r)\ell$ and we have that
		\[
		\dim V = \max_r \dim \{X \in \FF^{e \times h} \mid \matrixrank X = r\} + (h-r)\ell = \max_r er + (h-r)r + (h-r)\ell = \max_r f(r) + h\ell
		\]
		where $f(r) = r(\D-r)$ with $\D:=e+h-\ell$.
		Thus,
		\[\GR(T) = eh - \max_r f(r).\]
		Over the integers, the function $f$ attains its maximum at $\floor{\frac{\D}{2}}$ (and at $\ceil{\frac{\D}{2}}$), but this may be outside the interval $[0, e]$ that we want to maximize over. Recall that $e \le h \le \ell$. 
		Observe that if $\D \ge 0$ then 
		$e=\D+(\ell-h) \ge \D \ge 0$, 
		so $\floor{\frac{\D}{2}}$ is in the interval $[0,e]$.
		On the other hand, if $\D \le 0$ then $r(\D-r) \le 0 = f(0)$ for every $r \ge 0$, so the maximum of $f(r)$ over $r \in [0,e]$ is at the endpoint $r=0$.
		We thus find
		\begin{equation}
		\max_{0 \le r \le e} f(r) =
		\begin{cases}
		\floor{\frac{\D}{2}}\ceil{\frac{\D}{2}} = \floor{\frac{\D^2}{4}} & \textnormal{if } \D \ge 0,\\
		0 & \textnormal{otherwise}.
		\end{cases}
		\end{equation}
		This completes the proof.
	\end{proof}

	\begin{remark}
		\cref{stropt} gives the upper bound $\subrank(\langle m,m,m \rangle) \leq \bordersubrank(\langle m,m,m \rangle) = \ceil{\tfrac34 m^2}$ on the subrank of matrix multiplication $\subrank(\langle m,m,m \rangle)$. This improves the previously best known upper bound $\subrank(\langle m,m,m \rangle) \leq m^2 - m$
		from \cite[Equation 29]{christ2018tensor}.
	\end{remark}

	\begin{remark}\label{notsubm}
		Geometric rank $\GR$ is not sub-multiplicative under the tensor Kronecker product~$\otimes$. We give an example. The matrix multiplication tensor $\langle m,m,m\rangle$ can be written as the product $\langle m, m, m \rangle =  \langle m, 1, 1 \rangle \otimes \langle 1,m,1\rangle \otimes \langle 1, 1, m\rangle$ and $\GR(\langle m, 1, 1\rangle) = \GR(\langle 1, m, 1\rangle) = \GR(\langle 1, 1, m\rangle) = 1$
		whereas we have $\GR(\langle m, m, m\rangle) = \ceil{\tfrac34m^2}$ by \cref{stropt}.
	\end{remark}
	
	\begin{remark}\label{notqbar}
		Geometric rank $\GR$ is not the same as subrank $\subrank$ or border subrank $\bordersubrank$.
		For example, for the trilinear map $W(x_1, x_2, y_1, y_2, z_1, z_2) = x_1 y_1 z_2 + x_1 y_2 z_1 + x_2 y_1 z_1$ we find $\GR(W) = 2$ (see the example in the introduction), whereas $\subrank(W) = \bordersubrank(W) = 1$. The latter follows from the fact that $\asympsubrank(W) = 1.81...$~\cite{strassen1991degeneration}, where $\asympsubrank(T) \coloneqq \lim_{n\to\infty} \subrank(T^{\otimes n})^{1/n}$ is the asymptotic subrank of $T$, since $\bordersubrank(T) \leq \asympsubrank(T)$ \cite{strassen1987relative}. 
	\end{remark}
	
	\begin{remark}\label{remarkasymp}
		Geometric rank $\GR$ is not super-multiplicative under the tensor Kronecker product~$\otimes$. Here is an example.
		Let $\regularize{\SR}(T) \coloneqq \lim_{n\to\infty} \SR(T^{\otimes n})^{1/n}$ and let $\regularize{\GR}(T) \coloneqq \lim_{n\to\infty} \GR(T^{\otimes n})^{1/n}$, whenever these limits are defined.
		From the fact that $\subrank(T) \leq \GR(T) \leq \SR(T)$ and the fact that $\asympsubrank(W) = \regularize{\SR}(W) = 1.81...$ \cite{MR4495838}
		it follows that $\regularize{\GR}(W) = 1.81..$, whereas $\GR(W) = 2$.
		We conclude that~$\GR$ is not super-multiplicative.
		We have seen already in \cref{notsubm} that $\GR$ is not sub-multiplicative.
	\end{remark}
	
	\begin{remark}\label{notsr}
		Geometric rank $\GR$ is not the same as slice rank $\SR$.  For example, for the matrix multiplication tensor $\langle m,m,m\rangle$ we find that $\GR(\langle m, m, m\rangle) = \ceil{\tfrac34m^2}$ (\cref{stropt}), whereas it was known that $\SR(\langle m,m,m\rangle) = m^2$ \cite[Remark 4.9]{MR3631613}.
	\end{remark}

	\section{Geometric rank versus slice rank}\label{sec:moredirect}
	
	In \cref{sec:basicbounds} 
	we proved, by chaining the basic properties of geometric rank, that geometric rank is at most slice rank, that is, 
	$\GR(T) \leq \SR(T)$.
	What is the largest gap between $\GR(T)$ and~$\SR(T)$? Motivated by this question, and motivated by the analogous question for analytic rank instead of geometric rank that we discussed in the introduction, we give a direct proof of the inequality $\GR(T) \leq \SR(T)$.
	
	In fact, we prove a chain of inequalities $\GR(T) \leq \ZR(T) \leq \SR(T)$ where $\ZR(T)$ is defined as follows.
	We will henceforth use the following piece of notation for a tensor $T \in \F^{n_1 \times n_2 \times n_3}$;
	\begin{equation}\label{eq:VT-notation}
	\V(T) = \{(x,y) \in \FF^{n_1} \times \FF^{n_2} \mid \forall z \in \FF^{n_3} : T(x,y,z) = 0\} .
	\end{equation}
	Moreover, we use the following standard notation for the variety cut out by polynomials $f_1, \ldots, f_s$;
	\begin{equation}\label{eq:V-notation}
	\V(f_1, \ldots, f_s) = \{ x \mid f_1(x)=\cdots=f_s(x)=0 \}.
	\end{equation}
	Let $\FF[\b{x},\b{y}] = \FF[x_1, \ldots, x_{n_1}, y_1, \ldots, y_{n_2}]$
	and let $\FF[\b{x},\b{y},\b{z}] = \FF[x_1, \ldots, x_{n_1}, y_1, \ldots, y_{n_2}, z_1, \ldots, z_{n_3}]$.
	We denote by $\FF[\b{x},\b{y}]_{\{(0,1), (1,0), (1,1)\}} \subseteq \FF[\b{x},\b{y}]$ the subset of polynomials that are bi-homogeneous of bi-degree $(0,1)$, $(1,0)$ or $(1,1)$. That is, the set~$\FF[\b{x},\b{y}]_{\{(0,1), (1,0), (1,1)\}}$ contains the polynomials
	%
	in $\FF[x_1, \ldots, x_{n_1}]$ that are homogeneous of degree~1, and the polynomials in $\FF[y_1, \ldots, y_{n_2}]$ that are homogeneous of degree~1, and the polynomials in~$\FF[\b{x},\b{y}]$ that are homogeneous of degree $1$ in $x_1, \ldots, x_{n_1}$ and homogeneous of degree 1 in $y_1, \ldots, y_{n_2}$.
	For any tensor $T$ we define
	\[
	\ZR(T) = \min\bigl\{s \in \NN \mid \exists f_1,\ldots,f_s \in \FF[\b{x},\b{y}]_{\{(0,1), (1,0), (1,1)\}} : \V(f_1,\ldots,f_s) \sub \V(T) \bigr\}.
	\]

	\begin{theorem}\label{moredirect}
		Let $T$ be a tensor.
		Then $\GR(T) \le \ZR(T) \le \SR(T)$.
	\end{theorem}
	\begin{proof}
		We prove that $\ZR(T) \le \SR(T)$.
		Let $r=\SR(T)$.
		View $T$ as a polynomial $T \in \FF[\b{x},\b{y},\b{z}]$.
		We write $T = \sum_{i=1}^r T_i$ with $\SR(T_i)=1$ for every $i$.
		Then $T_i = f_ig_i$ for some $f_i \in \FF[\b{x},\b{y}]_{\{(0,1), (1,0), (1,1)\}}$ and $g_i \in \FF[\b{x},\b{y},\b{z}]$.
		We claim that $\V(f_1,\ldots,f_r) \sub \V(T)$. Indeed, if $(x,y) \in \V(f_1,\ldots,f_r)$, then $T_i(x,y,z)=0$ for every $i$ and every $z$, and therefore $T(x,y,z)=0$ for every $z$.
		We conclude that $\ZR(T) \le r = \SR(T)$.
		
		We prove that $\GR(T) \le \ZR(T)$.
		Let $s=\ZR(T)$.
		Then there are $f_1, \ldots, f_s \in \FF[\b{x},\b{y}]_{\{(0,1), (1,0), (1,1)\}}$ such that
		$\V(f_1,\ldots,f_s) \sub \V(T)$. 
		We have
		\begin{equation*}\label{eq:GR<SR}
		\GR(T) = \codim \V(T) \le \codim \V(f_1,\ldots,f_s) \le s = \ZR(T),
		\end{equation*}
		where the first inequality follows from the containment $\V(f_1,\ldots,f_s) \sub \V(T)$ which implies the inequality $\dim \V(f_1,\ldots,f_s) \le \dim \V(T)$.
	\end{proof}

	\section{Geometric rank as liminf of analytic rank}\label{sec:ARGR}

	Let $\PP \sub \NN$ denote the set of primes numbers.
	Recall that for a tensor $T$ over $\ZZ$ and a prime $p \in \PP$, we denote by $T_p$ the $3$-tensor over $\FF_p$ obtained by reducing all coefficients of $T$ modulo $p$.
	In this section we prove the following tight relationship between $\AR(T_p)$ and $\GR(T)$,
	
	\begin{theorem}\label{theo:liminf}
		For every tensor $T$ over $\ZZ$ we have 
		\[
		\liminf_{p \in \PP} \AR(T_p) = \GR(T).
		\]
	\end{theorem}
	
	The proof of Theorem~\ref{theo:liminf} is easily obtained from the following result, which determines the dimension of a variety using the number of rational points in prime finite fields.\footnote{Point counting has also been used to determine Betti numbers, as in~\cite{batyrev_1999}.} 
	
	\begin{lemma}\label{prop:dim-limsup}
		For every variety $V$ defined over $\ZZ$ we have
		\[\dim V = \limsup_{p \in \PP} \log_p |V(\FF_p)|.\] 
	\end{lemma}
	
	Indeed, Theorem~\ref{theo:liminf} would follow from Lemma~\ref{prop:dim-limsup} together with the identity~(\ref{eq:AR-identity}) showing that the analytic rank can be written in terms of the number of $\FF_p$-points of the algebraic variety $\V(T_p)$ (recall~(\ref{eq:VT-notation})): 
	for any tensor $T \in \ZZ^{n_1 \times n_2 \times n_3}$,
	\[
	\AR(T_p) = n_1 + n_2 -\log_p \abs[0]{\V(T_p)(\FF_p)}
	= n_1 + n_2 -\log_p \abs[0]{\V(T)(\FF_p)} .
	\]
	
	

	
	For the proof of Lemma~\ref{prop:dim-limsup} we will use several auxiliary results.
	The first result is a version of the Lang--Weil Theorem~\cite{LangWe54} (see~\cite{Tao12}, Corollary 4).
	
	\begin{theorem}[Lang--Weil bound, alternate form]\label{theo:LW}
		Let $\FF$ be a finite field.
		For every variety $V$ defined over~$\overline{\FF}$,
		$$|V(\FF)| = (c(V)+O_V(|\FF|^{-1/2}))|\FF|^{\dim V} $$
		where $c(V)$ is the number of top-dimensional irreducible components of $V$ defined over $\FF$.
	\end{theorem}
	
	We note that the notation $O_V(\cdot)$ represents a constant independent of $|\FF|$, instead depending only on the polynomials cutting out $V$ (their number, their degrees, and their number of variables).
	
	The following lemma can be obtained from a result of Cassels~\cite{cassels_1976} (also implicit in Lech~\cite{Lech_1953})
	on $p$-adic integers (by embedding the $p$-adic integers $\ZZ_p$ into the finite field $\FF_p$).

	\begin{lemma}\label{lemma:roots-modp}
		%
		For every finite set of algebraic integers $S$ there are infinitely many prime numbers $p$ for which there is a homomorphism from $\ZZ[S]$ to $\F_p$.
	\end{lemma}

	Finally, we will need the following two results which control the behavior of a variety under the application of a homomorphism on polynomials cutting it out.
	Henceforth, for a ring homomorphism $\psi \colon R \to R'$,
	if $f \in R[\b{x}]$ is a polynomial with coefficients in $R$ then $\psi(f) \in R'[\b{x}]$ is obtained by applying~$\psi$ on the coefficients of $f$; 
	moreover, for a subset $S \sub R$, denote $\psi(S) = \langle \psi(f) \,\vert\, f \in S \rangle \sub R'$,
	where $\langle \cdot \rangle$ stands for the generated ideal in the given ring containing the generating elements.

	
	\begin{theorem}[Bertini--Noether]\label{theo:BN}
		Let $f_1,\ldots,f_m \in R[\mathbf{x}]$, where $R$ is an integral domain, such that $X=\V(f_1,\ldots,f_m)$ is irreducible.
		There exists a nonzero $c \in R$ such that for every homomorphism $\psi \colon R \to \K$ into a field $\K$, if $\psi(c) \neq 0$ then $\V(\psi(f_1),\ldots,\psi(f_m))$ is irreducible of dimension $\dim X$.
	\end{theorem}
	
	\begin{lemma}\label{lemma:decomposition-homo}
		Let $\psi \colon R \to R'$ be a homomorphism of integral domains.
		If $V=\bigcup_i V_i$ is a finite union of varieties,
		$V=\V(S)$ and $V_i = \V(S_i)$ with $S,S_i \sub R[\b{x}]$, 
		then $\V(\psi(S)) = \bigcup_i \V(\psi(S_i))$.
		%
		%
	\end{lemma}
	
	Both results above can be obtained from scheme-theoretic results of Grothendieck (see Corollary~9.2.6.2 and Proposition~9.7.8 in~\cite{EGA-IV3}).
	Alternatively, Theorem~\ref{theo:BN} appears as Proposition~10.4.2 in~\cite{FieldArithmetic}, 
	and for Lemma~\ref{lemma:decomposition-homo} we include a proof below for completeness.
	%
	
	For a variety $V \sub \overline{\FF}^n$, 
	the \emph{ideal of $V$} is the ideal $\I(V) = \{f \in \overline{\FF}[x_1,\ldots,x_n] \mid f(p)=0 \, \forall p \in V\}$.
	For an ideal $I$ of a ring $R$, the radical of $I$ (in $R$) is the ideal $\sqrt{I} = \{f \in R \,\vert\, \exists n \in \N \colon f^n \in I \} \sub R$.
	We will use the following version of Hilbert's Nullstellensatz (see, e.g., Proposition 9.4.2 in~\cite{FieldArithmetic}).
	
	\begin{theorem}[Hilbert's Nullstellensatz]\label{theo:Hilbert}
		Let $\FF$ be any field.
		For every variety $V=\V(I)$, with $I$ an ideal of $\FF[\b{x}]$, we have
		$\I(V) \cap \FF[\b{x}] = \sqrt{I}$.
	\end{theorem}
	
	
	\begin{proof}[Proof of Lemma~\ref{lemma:decomposition-homo}]
		Note that $\psi(T) = \psi(\langle T \rangle)$ for a subset $T \sub R[\b{x}]$.
		Denote $I = \langle S \rangle$ and $I_i = \langle S_i \rangle$ (ideals of $R[\b{x}]$). 
		By Theorem~\ref{theo:Hilbert}, 
		$\sqrt{I} 
		= \sqrt{\prod_i I_i}$
		(as
		$\I(V) \cap R[\b{x}] 
		= \bigcap_i \I(V_i) \cap R[\b{x}] 
		= \bigcap_i \sqrt{I_i} = \sqrt{\bigcap_i I_i}$);		
		we need to prove
		$\sqrt{\psi(I)\overline{\KK}[\b{x}]} 
		= \sqrt{\prod_i \psi(I_i)\overline{\KK}[\b{x}]}$,
		where $\KK$ is the quotient field of $R'$.
		Let $\overline{\psi} \colon R \to \overline{\KK}$ be obtained from $\psi$ by embedding $R'$ in $\overline{\KK}$.
		Then
		\begin{center}
			$\sqrt{\prod_i \overline{\psi}(I_i)}
			= \sqrt{\overline{\psi}(\prod_i I_i)}
			= \sqrt{\overline{\psi}\big(\sqrt{\prod_i I_i}\big)}
			= \sqrt{\overline{\psi}\big(\sqrt{I}\big)}
			= \sqrt{\overline{\psi}(I)}$
		\end{center}	
		as we needed to prove, 
		where the second and last equalities follow from the identity
		$\sqrt{\phi(J)}=\sqrt{\phi(\sqrt{J})}$
		for an arbitrary ring homomorphism $\phi \colon R \to R_0$ and ideal $J$ of $R$---which we verify  next.
		For the containment $\supseteq$, if $p \in \sqrt{\phi(J)}$ then $p^n \in \phi(J) \sub \phi(\sqrt{J})$ for some $n\in\NN$, hence $p \in \sqrt{\phi(\sqrt{J})}$.
		For the reverse containment $\subseteq$, 
		if $p \in \sqrt{\phi(\sqrt{J})}$ then $p^n = \phi(f)$ for some $n \in \NN$ and $f \in R$, where $f^{k} \in J$ for some $k \in \NN$;
		thus, $p^{nk} = \phi(f)^k = \phi(f^k) \in \phi(J)$,
		that is, $p \in \sqrt{\phi(J)}$. 
		This completes the proof.
	\end{proof}

	Having collected the necessary auxiliary results, we now prove the main result of this section.
	
	\begin{proof}[Proof of Lemma~\ref{prop:dim-limsup}]
		For a homomorphism $\psi\colon\ZZ \to \KK$, we denote by $\psi(V)$
		the variety cut out by applying $\psi$ on the set---guaranteed by the statement---of polynomials over $\ZZ$ cutting out $V$.
		For a prime number $p$, let $f_p \colon \ZZ \to \FF_p$ be the mod-$p$ homomorphim.
		%
		First, we claim that $\dim f_p(V) = \dim V$ for all but finitely many $p$.
		For any such $p$, this would imply by Theorem~\ref{theo:LW} that
		$$|V(\FF_p)| = |f_p(V)(\FF_p)| \le O_V(p^{\dim f_p(V)}) = O_V(p^{\dim V}),$$
		which proves that $\limsup_{p \in \PP} \log_p |V(\FF_p)| \le \dim V$.
		The claim can be obtained from scheme-theoretic results of Grothendieck (see Corollary~9.2.6.2 in~\cite{EGA-IV3}), or alternatively from Theorem~\ref{theo:BN} as follows. 
		Let $A$ be the ring of algebraic integers, and for every prime $p$, let $g_p:A \to \overline{\FF_p}$ be an extension of $f_p$ to $A$.
		For any irreducible component $X$ of $V$ we have, 
		by Theorem~\ref{theo:BN} applied on $I:=\I(X) \cap A[\b{x}]$, that the variety $g_p(X):=\V(g_p(I))$ is irreducible of dimension $\dim X$, for all but finitely many $p$. 
		Let $V = \bigcup_i X_i$ be the decomposition of $V$ into irreducible components. By Lemma~\ref{lemma:decomposition-homo}, 
		$g_p(V) = \bigcup_i g_p(X_i)$ is the decomposition of $g_p(V)$ into irreducible components, for all but finitely many $p$,
		in which case $\dim g_p(V) = \max_i \dim g_p(X_i) = \max_i \dim X_i = \dim V$.
		Since $f_p(V)=g_p(V)$, our claim follows.
		
		
		Next, fix a top-dimensional irreducible component $X$ of $V$, and let $R$ be a ring such that $\I(X)$ is generated by polynomials over $R$. Note that $R$ can be taken to be an extension ring $R=\ZZ[S]$ for some finite set of algebraic numbers $S$ (to obtain $S$, take any finite generating set of $\I(X)$, clear denominators, and let $S$ be set of all their coefficients). 
		By Lemma~\ref{lemma:roots-modp}, there is a homomorphism $h_p:R \to \FF_p$ for infinitely many $p$.
		By Theorem~\ref{theo:BN} applied on $I:=\I(X) \cap R[\b{x}]$, the variety $h_p(X):=\V(h_p(I))$ is irreducible of dimension $\dim X$ $(= \dim V)$ for infinitely many $p$.
		Importantly, Lemma~\ref{lemma:decomposition-homo} implies $h_p(X) \sub h_p(V)$. 
		For any $p$ as above,
		apply Theorem~\ref{theo:LW} to obtain
		$$|V(\FF_p)| = |h_p(V)(\FF_p)| \ge |h_p(X)(\FF_p)| \ge p^{\dim V}(1 - o_V(1)),$$
		which proves that $\limsup_{p \in \PP} \log_p |V(\FF_p)| \ge \dim V$.
		This completes the proof.
	\end{proof}

	\section{Open problems}

	We discuss several open problems about geometric rank.

	\begin{itemize}
		\item  As discussed in \cref{sec:gr}, while the work of Koiran \cite{646091} gives bounds on the computational complexity of computing the dimension of general algebraic varieties, it remains open to determine the computational complexity of computing the geometric rank (i.e.~computing the dimension of a \emph{bilinear} variety).
		\item In \cref{stropt} we determined the precise value of the \emph{border} subrank $\bordersubrank$ of the matrix multiplication tensors. It remains open to determine the subrank $\subrank$ of the matrix multiplication tensors. 
		\item Related to the previous point it is natural to ask what can generally be said about the relation between (border) subrank and geometric rank.
		\item In \cref{theo:liminf} we have shown that $\liminf_p \AR(T_p) = \GR(T)$ for any tensor $T$ over $\ZZ$, where~$T_p$ denotes the tensor over $\FF_p$ obtained by reducing all coefficients of $T$ modulo $p$.
		What is the largest gap between $\liminf_p \AR(T_p)$ and $\limsup_p \AR(T_p)$?
		\item Related to the previous point, can \cref{theo:liminf} be extended to tensors $T$ defined over finite fields rather than over $\ZZ$?
		Indeed, one can extend the definition of analytic rank to non-prime finite fields, by evaluating the bias of a character applied to the tensor. Does an analogous result to \cref{theo:liminf} hold there as well?
	\end{itemize}


\section*{Acknowledgments} 
We would like to thank Avi Wigderson for helpful conversations, and the anonymous reviewers for their careful reading and comments.





\raggedright
\bibliographystyle{alphaurl} 
\bibliography{all}

\newcommand{\etalchar}[1]{$^{#1}$}
\providecommand{\noopsort}[1]{}
\begin{thebibliography}{BCC{\etalchar{+}}17b}

\bibitem[AFLG15]{MR3388238}
Andris Ambainis, Yuval Filmus, and Fran\c{c}ois Le~Gall.
\newblock Fast matrix multiplication: limitations of the
  {C}oppersmith-{W}inograd method (extended abstract).
\newblock In {\em {P}roceedings of the 47th Annual {ACM} {S}ymposium on
  {T}heory of {C}omputing ({STOC} 2015)}, pages 585--593, 2015.
\newblock \href {http://arxiv.org/abs/1411.5414} {\path{arXiv:1411.5414}},
  \href {https://doi.org/10.1145/2746539.2746554}
  {\path{doi:10.1145/2746539.2746554}}.

\bibitem[Alm19]{DBLP:conf/coco/Alman19}
Josh Alman.
\newblock Limits on the universal method for matrix multiplication.
\newblock In {\em Proceedings of the 34th Computational Complexity Conference
  ({CCC} 2019)}, pages 12:1--12:24, 2019.
\newblock \href {http://arxiv.org/abs/1812.08731} {\path{arXiv:1812.08731}},
  \href {https://doi.org/10.4230/LIPIcs.CCC.2019.12}
  {\path{doi:10.4230/LIPIcs.CCC.2019.12}}.

\bibitem[AW18]{8555139}
Josh Alman and Virginia~Vassilevska Williams.
\newblock Limits on all known (and some unknown) approaches to matrix
  multiplication.
\newblock In {\em Proceedings of the 59th Annual IEEE Symposium on Foundations
  of Computer Science ({FOCS} 2018)}, pages 580--591, 2018.
\newblock \href {http://arxiv.org/abs/1810.08671} {\path{arXiv:1810.08671}},
  \href {https://doi.org/10.1109/FOCS.2018.00061}
  {\path{doi:10.1109/FOCS.2018.00061}}.

\bibitem[Bat99]{batyrev_1999}
Victor~V. Batyrev.
\newblock {\em Birational Calabi–Yau n-folds have equal Betti numbers}, page
  1–12.
\newblock London Mathematical Society Lecture Note Series. Cambridge University
  Press, 1999.
\newblock \href {https://doi.org/10.1017/CBO9780511721540.002}
  {\path{doi:10.1017/CBO9780511721540.002}}.

\bibitem[BCC{\etalchar{+}}17a]{MR3631613}
Jonah Blasiak, Thomas Church, Henry Cohn, Joshua~A. Grochow, Eric Naslund,
  William~F. Sawin, and Chris Umans.
\newblock On cap sets and the group-theoretic approach to matrix
  multiplication.
\newblock {\em Discrete Anal.}, 2017.
\newblock \href {http://arxiv.org/abs/1605.06702} {\path{arXiv:1605.06702}},
  \href {https://doi.org/10.19086/da.1245} {\path{doi:10.19086/da.1245}}.

\bibitem[BCC{\etalchar{+}}17b]{blasiak2017groups}
Jonah Blasiak, Thomas Church, Henry Cohn, Joshua~A Grochow, and Chris Umans.
\newblock Which groups are amenable to proving exponent two for matrix
  multiplication?, 2017.
\newblock \href {http://arxiv.org/abs/1712.02302} {\path{arXiv:1712.02302}}.

\bibitem[BCS97]{burgisser1997algebraic}
Peter B{\"u}rgisser, Michael Clausen, and M.~Amin Shokrollahi.
\newblock {\em Algebraic complexity theory}, volume 315 of {\em Grundlehren
  Math. Wiss.}
\newblock Springer-Verlag, Berlin, 1997.
\newblock \href {https://doi.org/10.1007/978-3-662-03338-8}
  {\path{doi:10.1007/978-3-662-03338-8}}.

\bibitem[BHH{\etalchar{+}}20]{bhrushundi_et_al:LIPIcs:2020:12632}
Abhishek Bhrushundi, Prahladh Harsha, Pooya Hatami, Swastik Kopparty, and
  Mrinal Kumar.
\newblock {On Multilinear Forms: Bias, Correlation, and Tensor Rank}.
\newblock In {\em Approximation, Randomization, and Combinatorial Optimization.
  Algorithms and Techniques (APPROX/RANDOM 2020)}, volume 176, pages
  29:1--29:23, 2020.
\newblock \href {http://arxiv.org/abs/1804.09124} {\path{arXiv:1804.09124}},
  \href {https://doi.org/10.4230/LIPIcs.APPROX/RANDOM.2020.29}
  {\path{doi:10.4230/LIPIcs.APPROX/RANDOM.2020.29}}.

\bibitem[BIL{\etalchar{+}}19]{blser2019variety}
Markus Bläser, Christian Ikenmeyer, Vladimir Lysikov, Anurag Pandey, and
  Frank-Olaf Schreyer.
\newblock Variety membership testing, algebraic natural proofs, and geometric
  complexity theory, 2019.
\newblock \href {http://arxiv.org/abs/1911.02534} {\path{arXiv:1911.02534}}.

\bibitem[BL15]{bhowmick2015bias}
Abhishek Bhowmick and Shachar Lovett.
\newblock Bias vs structure of polynomials in large fields, and applications in
  effective algebraic geometry and coding theory, 2015.
\newblock \href {http://arxiv.org/abs/1506.02047} {\path{arXiv:1506.02047}}.

\bibitem[Bri21]{briet2019subspaces}
Jop Bri\"{e}t.
\newblock Subspaces of tensors with high analytic rank.
\newblock {\em Online J. Anal. Comb.}, 16(6), 2021.
\newblock \href {http://arxiv.org/abs/1908.04169} {\path{arXiv:1908.04169}}.

\bibitem[Cas76]{cassels_1976}
J.W.S. Cassels.
\newblock An embedding theorem for fields.
\newblock {\em Bulletin of the Australian Mathematical Society},
  14(2):193–198, 1976.
\newblock \href {https://doi.org/10.1017/S000497270002503X}
  {\path{doi:10.1017/S000497270002503X}}.

\bibitem[CLVW20]{christ2018tensor}
Matthias Christandl, Angelo Lucia, Péter Vrana, and Albert~H. Werner.
\newblock {Tensor network representations from the geometry of entangled
  states}.
\newblock {\em SciPost Phys.}, 9:042, 2020.
\newblock \href {http://arxiv.org/abs/1809.08185v2}
  {\path{arXiv:1809.08185v2}}, \href
  {https://doi.org/10.21468/SciPostPhys.9.3.042}
  {\path{doi:10.21468/SciPostPhys.9.3.042}}.

\bibitem[CM21]{cohen2021partition}
Alex Cohen and Guy Moshkovitz.
\newblock Partition and analytic rank are equivalent over large fields, 2021.
\newblock \href {http://arxiv.org/abs/2102.10509} {\path{arXiv:2102.10509}}.

\bibitem[CM22]{MR4492184}
Alex Cohen and Guy Moshkovitz.
\newblock Structure vs.~randomness for bilinear maps.
\newblock {\em Discrete Anal.}, 12, 2022.
\newblock \href {https://doi.org/10.19086/da.38587}
  {\path{doi:10.19086/da.38587}}.

\bibitem[CVZ21]{MR4322894}
Matthias Christandl, P\'{e}ter Vrana, and Jeroen Zuiddam.
\newblock Barriers for fast matrix multiplication from irreversibility.
\newblock {\em Theory Comput.}, 17(2), 2021.
\newblock \href {http://arxiv.org/abs/1812.06952} {\path{arXiv:1812.06952}},
  \href {https://doi.org/10.4086/toc.2021.v017a002}
  {\path{doi:10.4086/toc.2021.v017a002}}.

\bibitem[CVZ23]{MR4495838}
Matthias Christandl, P\'{e}ter Vrana, and Jeroen Zuiddam.
\newblock Universal points in the asymptotic spectrum of tensors.
\newblock {\em J. Amer. Math. Soc.}, 36(1):31--79, 2023.
\newblock \href {http://arxiv.org/abs/1709.07851} {\path{arXiv:1709.07851}},
  \href {https://doi.org/10.1090/jams/996} {\path{doi:10.1090/jams/996}}.

\bibitem[CW90]{MR1056627}
Don Coppersmith and Shmuel Winograd.
\newblock Matrix multiplication via arithmetic progressions.
\newblock {\em J. Symbolic Comput.}, 9(3):251--280, 1990.
\newblock \href {https://doi.org/10.1016/S0747-7171(08)80013-2}
  {\path{doi:10.1016/S0747-7171(08)80013-2}}.

\bibitem[EG17]{MR3583358}
Jordan~S. Ellenberg and Dion Gijswijt.
\newblock On large subsets of {$\mathbb{F}^n_q$} with no three-term arithmetic
  progression.
\newblock {\em Ann. of Math. (2)}, 185(1):339--343, 2017.
\newblock \href {https://doi.org/10.4007/annals.2017.185.1.8}
  {\path{doi:10.4007/annals.2017.185.1.8}}.

\bibitem[FGL19]{filmus2019high}
Yuval Filmus, Konstantin Golubev, and Noam Lifshitz.
\newblock High dimensional {Hoffman} bound and applications in extremal
  combinatorics, 2019.
\newblock \href {http://arxiv.org/abs/1911.02297} {\path{arXiv:1911.02297}}.

\bibitem[FJ05]{FieldArithmetic}
Michael~D. Fried and Moshe Jarden.
\newblock {\em Field arithmetic}, volume~11 of {\em Ergebnisse der Mathematik
  und ihrer Grenzgebiete. 3. Folge. A Series of Modern Surveys in Mathematics}.
\newblock Springer-Verlag, Berlin, second edition, 2005.

\bibitem[Gen22]{https://doi.org/10.48550/arxiv.2201.03615}
Runshi Geng.
\newblock Geometric rank and linear determinantal varieties, 2022.
\newblock \href {https://doi.org/10.48550/ARXIV.2201.03615}
  {\path{doi:10.48550/ARXIV.2201.03615}}.

\bibitem[GL22]{geng2021geometry}
Runshi Geng and Joseph~M. Landsberg.
\newblock {On the geometry of geometric rank}.
\newblock {\em Algebra \& Number Theory}, 16(5):1141--1160, 2022.
\newblock \href {http://arxiv.org/abs/2012.04679} {\path{arXiv:2012.04679}},
  \href {https://doi.org/10.2140/ant.2022.16.1141}
  {\path{doi:10.2140/ant.2022.16.1141}}.

\bibitem[Gro66]{EGA-IV3}
Alexander Grothendieck.
\newblock {\'E}l\'ements de g\'eom\'etrie alg\'ebrique: {IV}. {{\'E}}tude
  locale des sch\'emas et des morphismes de sch\'emas, troisi\`eme partie.
\newblock {\em Publications Math\'ematiques de l'IH\'ES}, 28:5--255, 1966.
\newblock URL: \url{http://www.numdam.org/item/PMIHES_1966__28__5_0/}.

\bibitem[GS]{M2}
Daniel~R. Grayson and Michael~E. Stillman.
\newblock Macaulay2, a software system for research in algebraic geometry.
\newblock Available at \url{http://www.math.uiuc.edu/Macaulay2/}.

\bibitem[GT09]{DBLP:journals/cdm/GreenT09}
Ben~Joseph Green and Terence Tao.
\newblock The distribution of polynomials over finite fields, with applications
  to the gowers norms.
\newblock {\em Contributions Discret. Math.}, 4(2), 2009.
\newblock URL: \url{http://cdm.ucalgary.ca/cdm/index.php/cdm/article/view/133}.

\bibitem[GW11]{MR2773103}
W.~T. Gowers and J.~Wolf.
\newblock Linear forms and higher-degree uniformity for functions
  on~{$\mathbb{F}^n_p$}.
\newblock {\em Geom. Funct. Anal.}, 21(1):36--69, 2011.
\newblock \href {https://doi.org/10.1007/s00039-010-0106-3}
  {\path{doi:10.1007/s00039-010-0106-3}}.

\bibitem[Har92]{harris2013algebraic}
Joe Harris.
\newblock {\em Algebraic geometry: a first course}, volume 133.
\newblock Springer Science \& Business Media, 1992.

\bibitem[Jan18]{janzer2018low}
Oliver Janzer.
\newblock Low analytic rank implies low partition rank for tensors, 2018.
\newblock \href {http://arxiv.org/abs/1809.10931} {\path{arXiv:1809.10931}}.

\bibitem[Jan20]{janzer2019polynomial}
Oliver Janzer.
\newblock Polynomial bound for the partition rank vs the analytic rank of
  tensors.
\newblock {\em Discrete Anal.}, 2020.
\newblock \href {http://arxiv.org/abs/1902.11207} {\path{arXiv:1902.11207}}.

\bibitem[KL08]{DBLP:conf/focs/KaufmanL08}
Tali Kaufman and Shachar Lovett.
\newblock Worst case to average case reductions for polynomials.
\newblock In {\em 49th Annual {IEEE} Symposium on Foundations of Computer
  Science, {FOCS} 2008, October 25-28, 2008, Philadelphia, PA, {USA}}, pages
  166--175. {IEEE} Computer Society, 2008.
\newblock \href {https://doi.org/10.1109/FOCS.2008.17}
  {\path{doi:10.1109/FOCS.2008.17}}.

\bibitem[Koi97]{646091}
Pascal Koiran.
\newblock Randomized and deterministic algorithms for the dimension of
  algebraic varieties.
\newblock In {\em Proceedings 38th Annual Symposium on Foundations of Computer
  Science}, pages 36--45, Oct 1997.
\newblock \href {https://doi.org/10.1109/SFCS.1997.646091}
  {\path{doi:10.1109/SFCS.1997.646091}}.

\bibitem[Lec53]{Lech_1953}
Christer Lech.
\newblock {A note on recurring series}.
\newblock {\em Arkiv för Matematik}, 2(5):417 -- 421, 1953.
\newblock \href {https://doi.org/10.1007/BF02590997}
  {\path{doi:10.1007/BF02590997}}.

\bibitem[LG14]{le2014powers}
Fran{\c{c}}ois Le~Gall.
\newblock Powers of tensors and fast matrix multiplication.
\newblock In {\em {P}roceedings of the 39th {I}nternational {S}ymposium on
  {S}ymbolic and {A}lgebraic {C}omputation ({ISSAC} 2014)}, pages 296--303,
  2014.
\newblock \href {http://arxiv.org/abs/1401.7714} {\path{arXiv:1401.7714}},
  \href {https://doi.org/10.1145/2608628.2608664}
  {\path{doi:10.1145/2608628.2608664}}.

\bibitem[Lov19]{lovett2019analytic}
Shachar Lovett.
\newblock The analytic rank of tensors and its applications.
\newblock {\em Discrete Anal.}, 2019.
\newblock \href {http://arxiv.org/abs/1806.09179} {\path{arXiv:1806.09179}},
  \href {https://doi.org/10.19086/da.8654} {\path{doi:10.19086/da.8654}}.

\bibitem[LW54]{LangWe54}
Serge Lang and Andr\'{e} Weil.
\newblock Number of points of varieties in finite fields.
\newblock {\em Amer. J. Math.}, 76:819--827, 1954.
\newblock \href {https://doi.org/10.2307/2372655} {\path{doi:10.2307/2372655}}.

\bibitem[Mil19]{milicevic2019polynomial}
Luka Mili{\'c}evi{\'c}.
\newblock Polynomial bound for partition rank in terms of analytic rank.
\newblock {\em Geom.~Funct.~Anal.}, 29(5):1503--1530, 2019.
\newblock \href {http://arxiv.org/abs/1902.09830} {\path{arXiv:1902.09830}},
  \href {https://doi.org/10.1007/s00039-019-00505-4}
  {\path{doi:10.1007/s00039-019-00505-4}}.

\bibitem[Nas20]{NASLUND2020105190}
Eric Naslund.
\newblock The partition rank of a tensor and $k$-right corners in
  {$\mathbb{F}_q^n$}.
\newblock {\em Journal of Combinatorial Theory, Series A}, 174:105190, 2020.
\newblock \href {https://doi.org/10.1016/j.jcta.2019.105190}
  {\path{doi:10.1016/j.jcta.2019.105190}}.

\bibitem[NS17]{naslund2017upper}
Eric Naslund and Will Sawin.
\newblock Upper bounds for sunflower-free sets.
\newblock {\em Forum Math.\ Sigma}, 5:e15, 2017.
\newblock \href {http://arxiv.org/abs/1606.09575} {\path{arXiv:1606.09575}},
  \href {https://doi.org/10.1017/fms.2017.12} {\path{doi:10.1017/fms.2017.12}}.

\bibitem[Sag17]{sagemath}
Sage.
\newblock {\em {S}ageMath, the {S}age {M}athematics {S}oftware {S}ystem
  ({V}ersion 8.1)}, 2017.
\newblock {\tt https://www.sagemath.org}.

\bibitem[Sch84]{MR774104}
Wolfgang~M. Schmidt.
\newblock Bounds for exponential sums.
\newblock {\em Acta Arith.}, 44(3):281--297, 1984.
\newblock \href {https://doi.org/10.4064/aa-44-3-281-297}
  {\path{doi:10.4064/aa-44-3-281-297}}.

\bibitem[Sch85]{MR781588}
Wolfgang~M. Schmidt.
\newblock The density of integer points on homogeneous varieties.
\newblock {\em Acta Math.}, 154(3-4):243--296, 1985.
\newblock \href {https://doi.org/10.1007/BF02392473}
  {\path{doi:10.1007/BF02392473}}.

\bibitem[Str87]{strassen1987relative}
Volker Strassen.
\newblock Relative bilinear complexity and matrix multiplication.
\newblock {\em J. Reine Angew. Math.}, 375/376:406--443, 1987.
\newblock \href {https://doi.org/10.1515/crll.1987.375-376.406}
  {\path{doi:10.1515/crll.1987.375-376.406}}.

\bibitem[Str88]{strassen1988asymptotic}
Volker Strassen.
\newblock The asymptotic spectrum of tensors.
\newblock {\em J. Reine Angew. Math.}, 384:102--152, 1988.
\newblock \href {https://doi.org/10.1515/crll.1988.384.102}
  {\path{doi:10.1515/crll.1988.384.102}}.

\bibitem[Str91]{strassen1991degeneration}
Volker Strassen.
\newblock Degeneration and complexity of bilinear maps: some asymptotic
  spectra.
\newblock {\em J. Reine Angew. Math.}, 413:127--180, 1991.
\newblock \href {https://doi.org/10.1515/crll.1991.413.127}
  {\path{doi:10.1515/crll.1991.413.127}}.

\bibitem[Tao12]{Tao12}
Terence Tao.
\newblock The {L}ang-{W}eil bound, 2012.
\newblock URL:
  \url{https://terrytao.wordpress.com/2012/08/31/the-lang-weil-bound}.

\bibitem[Tao16]{tao}
Terence Tao.
\newblock A symmetric formulation of the
  {C}root-{L}ev-{P}ach-{E}llenberg-{G}ijswijt capset bound.
\newblock \url{https://terrytao.wordpress.com}, 2016.
\newblock URL:
  \url{https://terrytao.wordpress.com/2016/05/18/a-symmetric-formulation-of-the-croot-lev-pach-ellenberg-gijswijt-capset-bound}.

\end{thebibliography}


\begin{dajauthors}
\begin{authorinfo}[kopp]
  Swastik Kopparty\\
  Associate Professor\\
  University of Toronto\\
  Toronto, Canada\\
  swastik\imagedot{}kopparty\imageat{}utoronto\imagedot{}ca \\
  \url{https://www.math.toronto.edu/swastik/}
\end{authorinfo}
\begin{authorinfo}[moshk]
  Guy Moshkovitz\\
  Assistant Professor\\
  City University of New York (Baruch College and the Graduate Center)\\
  New York City, USA\\
  guymoshkov\imageat{}gmail\imagedot{}com \\
  \url{https://sites.google.com/view/guy-moshkovitz/}
\end{authorinfo}
\begin{authorinfo}[zuid]
  Jeroen Zuiddam\\
  Assistant Professor\\
  University of Amsterdam\\
  Amsterdam, Netherlands\\
  j\imagedot{}zuiddam\imageat{}uva\imagedot{}nl\\
  \url{https://staff.fnwi.uva.nl/j.zuiddam/}
\end{authorinfo}
\end{dajauthors}

\end{document}